\documentclass[11pt]{article}

\usepackage[margin=1in]{geometry}
\usepackage{amsmath,amsthm,amsfonts}
\usepackage[square, numbers]{natbib}

\title{Optimizing for Fairness in Generalized Kidney Exchange: \\ Theory and Computations}
\author{Claire S. Chang\thanks{Cornell University, \texttt{\{csc258, amk445, david.shmoys\}@cornell.edu}} \and Arin Khare\footnotemark[1] \and David B. Shmoys\footnotemark[1] \thanks{This research was partially funded by a grant from the National Science Foundation DMS-2230023.}}
\date{}

\newtheorem{lemma}{Lemma}
\newtheorem{proposition}{Proposition}
\newtheorem{theorem}{Theorem}

\usepackage{multicol}
\usepackage{hyperref}
\usepackage{mathrsfs}
\usepackage{bbm}

\usepackage{subcaption}
\usepackage{algorithm}
\usepackage[noend]{algpseudocode}

\usepackage{xargs}

\usepackage{tikz}
\usetikzlibrary{positioning,fit,backgrounds,arrows.meta}

\usepackage{thmtools,thm-restate}

\newcommand{\Dtilde}{\widetilde{D}(G)}
\renewcommand\vec{\mathbf}

\begin{document}
\maketitle

\begin{abstract}
The seminal work of Roth, S\"onmez, \& 
\"Unver shows that the Edmonds-Gallai structure theorem for non-bipartite matching can be leveraged to yield a randomized algorithm to match patient-donor pairs in kidney exchange with extraordinarily strong properties. This breakthrough led to randomized polynomial-time algorithms to find a maximum-cardinality matching maximizing individual fairness objectives---measured by the probability that nodes are matched---such as Nash social welfare. But the exchanges allowed in practice 
go beyond cardinality matching, generalizing to weighted variants and allowing structures such as paths and 3-cycles.
We show that strongly polynomial algorithms guaranteeing the same fairness properties can be obtained in weighted settings for matching and 2-paths. While even maximum cardinality coverage with cycles and paths of length at least three is NP-hard, we provide a general result showing that {\it any} optimization subroutine (for whichever structure is allowed) can be bootstrapped using a polynomial number of calls to yield a mechanism that has analogous fairness properties to those obtained for matching. We complement these theoretical results with computational results, both on well-studied synthetic data-sets and on samples drawn from real data, that demonstrate the striking advantages of adding fairness considerations to more general kidney-exchange mechanisms.
\end{abstract}

\section{Introduction}
To hasten the process for kidney transplants (with over 90,000 on the waiting list as of 2025 and most people waiting for deceased donors for 3-5 years \citep{UNOS}), waiting list patients can look for living donors. However, many patients have willing but incompatible donors. 
To tackle this issue, kidney paired donation programs, such as the National Kidney Registry (NKR), the Alliance for Paired Kidney Donation (APKD), and UNOS's kidney paired donation program in the United States, have been established in which incompatible patient-donor pairs can exchange donors with other incompatible pairs to find compatible matches. These exchanges are typically limited to pairs or triples of patient-donor pairs due to logistical and medical constraints: surgeries must be performed simultaneously to prevent reneging, so each pair that is included requires an additional operating room and surgical team. Another option is to include altruistic or nondirected donors (NDDs), who are willing to donate a kidney to any compatible patient. NDDs can start chains of donations that do not need to be performed simultaneously, allowing for longer chains and more transplants.
The organizations use algorithms based on integer programming to match patient-donor pairs and NDDs based on compatibility and other factors \citep{Roth_Sonmez_Utku_Unver_2005,Roth_Sonmez_M_2005,Anderson_Ashlagi_Gamarnik_Rees_Roth_Sonmez_Unver_2015,Ashlagi_Roth_2021}. However, an integer program yields only one solution, and it is left up to the whims of solvers to determine which pairs are included. Instead, we wish to randomize over the set of optimal solutions so that pairs can be included in an exchange with probabilities aimed to optimize a specified fairness function.

The kidney exchange problem can be modeled as a directed graph where vertices represent patient-donor pairs or NDDs, and there is a directed edge from vertex A to vertex B if the donor of A is compatible with the patient of B. There may be weights on the arcs which represent, for example, the quality of the matches. An exchange corresponds to a collection of vertex-disjoint cycles and paths in this graph. Typically, the length of the cycle or path is bounded to be at most an $\ell$-cycle or $\ell$-path, where $\ell$ is the number of arcs or exchanges made. In the context of exchanges with only 2-cycles, we can also view the graph as undirected with edges representing mutual compatibility, and exchanges as matchings. Two possible goals for this problem include 
\begin{enumerate}
    \item[(G1)] finding a single packing that optimizes some objective or satisfies some constraints, or 
    \item[(G2)] finding a lottery over the set of optimal packings that further optimizes some fairness criterion.
\end{enumerate}
The fairness criteria we consider will be a function of the probabilities with which each vertex is included in an exchange in the lottery. For instance, we may wish to maximize the smallest probability any vertex is included (maximin), or more generally lexicographically maximize the probabilities sorted in nondecreasing order (leximin). We can also try to maximize the sum of the probabilities (utilitarian, although this is equivalent to constraining to maximum-cardinality packings) or their product (Nash social welfare). 

One of the seminal works on paired kidney exchange by \citeauthor{Roth_Sonmez_Utku_Unver_2005} focused on exchanges including only pairs of patient-donor pairs with dichotomous valuations (i.e., an unweighted undirected nonbipartite graph). Their goal falls into (G2): constructing a lottery over the set of all maximum-cardinality matchings with the goal of leximin fairness. This work builds on earlier work of \citeauthor{Bogomolnaia_Moulin_2004}, who explained how to find the leximin solution for bipartite graphs. By leveraging the Edmonds-Gallai decomposition \citep{Edmonds_1965}, \citeauthor{Roth_Sonmez_Utku_Unver_2005} were able to extend the algorithm to the nonbipartite setting. Furthermore, they proved that the leximin solution is Lorenz dominant over the space of maximum-cardinality exchanges in this space, so it also maximizes the Nash social welfare and any other concave function of the probabilities.
Polynomial-time algorithms for finding the leximin solution in nonbipartite graphs were later presented in \cite{Garcia-Soriano_Bonchi_2020,Li_Liu_Huang_Tang_2014} using results on edge coloring and parametric flow.

Unfortunately, allowing much more than 2-cycles makes the problem much more difficult. A later stream of work focuses on (G1), and considers the problem of finding a single packing of length at most $\ell$-cycles or $\ell$-paths, as well as incorporating edge weights. \citeauthor{Brewster_Hell_Pantel_Rizzi_Yeo_2003}  present a polynomial-time algorithm for packing 1-paths and 2-paths, but show that the problem is NP-complete for almost any other family of graphs. However, for any constant $\ell \geq 3$, finding a maximum-cardinality packing of $\ell$-cycles, even in undirected graphs, is NP-complete and APX-complete \citep{Abraham_Blum_Sandholm_2007,Biro_Manlove_Rizzi_2009}. 
Thus as mentioned earlier, the focus has been on computational approaches involving integer programming formulations \citep{Roth_Sonmez_Utku_Unver_2005,Roth_Sonmez_M_2005,Anderson_Ashlagi_Gamarnik_Rees_Roth_Sonmez_Unver_2015,Anderson_Ashlagi_Gamarnik_Roth_2015,Ashlagi_Roth_2021}. 
Other goals for (G1) arise from economics, where researchers have focused on incentive compatibility for hospitals to ensure truthful reporting of all patients \citep{Ashlagi_Roth_2011,Ashlagi_Fischer_Kash_Procaccia_2015}.

Since (G2) is even more challenging than (G1), it is no surprise that literature in this area for the more general kidney exchange graph has also been focused on computational approaches for these more general settings.
Integer programming techniques with constraint programming, and more successfully column generation, can be used to construct lotteries over the set of all optimal solutions \citep{Farnadi_St-Arnaud_Babaki_Carvalho_2021,Demeulemeester_Goossens_Hermans_Leus_2025}. This approach has found success not only in kidney exchange, but also in other areas such as selecting citizens' assemblies \citep{Flanigan_Golz_Gupta_Hennig_Procaccia_2021}. Besides the computational hardness, there is also tension between the utilitarian and fairness objectives when we move past matching. In matching, it is well known that any non-isolated node can be included in some maximum matching. In the kidney exchange literature, \citeauthor{Dickerson_Procaccia_Sandholm_2014} target highly sensitized patients that are hard to match, and explore the ``price of fairness'' to see what loss in efficiency results from prioritizing these patients, and \citeauthor{McElfresh_Dickerson_2018} propose a hybrid mechanism that balances the two objectives.

\paragraph{Our contributions.} We show that given an algorithm for (G1), we can solve (G2) by finding an almost-optimal solution for any concave fairness metric in a polynomial number of calls to that algorithm using the ellipsoid method. Thus any properties of the underlying packing problem (e.g., incentive compatibility, computational complexity) carry over to the fairness setting. We focus on two properties in particular: the computational complexity and the ``coverage loss,'' the maximum decrease in the number of vertices covered compared to a maximum-cardinality packing in order to maximize individual fairness objectives or constraints. 
We also show that analogous results to \cite{Roth_Sonmez_Utku_Unver_2005} extend to more general settings. Specifically, adding node or edge weights and finding a lottery of maximum-weight matchings subject to a cardinality constraint still allows for a strongly polynomial 
algorithm for finding the Lorenz dominant solution. We also propose a strongly polynomial 
algorithm for the case of 2-cycle and 2-path packings, whereas with 3-paths the problem becomes NP-complete.
More generally, we show that with any family of packings that are downward-closed, there still exists a unique Lorenz dominant solution over fixed cardinality packings.
Finally, we study the price of fairness in the context of individuals. 
To complement these theoretical results, we examine the impact of the lotteries that result from our algorithmic insights. Using \cite{Anderson_Ashlagi_Gamarnik_Roth_2015} as a baseline measure, we demonstrate that our approach yields significant fairness improvements that can be achieved with small losses in the number of covered vertices. In a dynamic setting, we show that optimizing for fairness can decrease waiting times in the system, and this is essentially matched by a heuristic that merely shuffles patient-donor pairs before computing exchanges.

\section{Model}
A \emph{kidney exchange problem (KEP) instance} consists of a directed compatibility graph $G=(V,E)$ with $n=|V|$, and each node is associated with biological data such as blood types and antibody information. A solution consists of a packing $C \subseteq E$ of cycles and paths in the compatibility graph.
We will consider multiple variants of feasibility conditions. In general, we will consider a set of acceptable packings $\mathscr{C}$, and concave monotone-increasing individual fairness function $F(\vec{q})$, where $\vec{q} \in [0,1]^n$ represents the probability each patient-donor pair is included in the exchange. In the graph, vertices represent patient-donor pairs and an edge from $v_1$ to $v_2$ indicates that the recipient of $v_1$ is compatible with the donor of $v_2$. We assume that there are no self-loops (patient-donor pairs are incompatible with each other). We also assume that $G$ has no perfect packing, and each vertex is included in at least one packing of $\mathscr{C}$. Otherwise, we return a perfect packing or consider the subgraph induced by the vertices that are included in at least one packing, respectively. Some examples of acceptable packings we consider include maximum cardinality 2-cycles (i.e. matching), maximum weight 2-cycles, maximum cardinality 2-paths, and maximum cardinality 2- and 3-cycles and unbounded length paths.
We are interested in finding a constant size subset of $\mathscr{C}$ and a probability distribution $\vec{p}$ over those packings that maximizes $F: [0,1]^n \rightarrow [0,1]$. A higher value of $F$ indicates a more fair lottery over packings. 
Hence we are interested in the following optimization problem:
\begin{equation}
\label{lp:fairness}
    \max_{\vec{q} \in \textsf{marginal}(\mathscr{C})} F(\vec{q}),
\tag{FAIR}
\end{equation}

where 
$$
\textsf{marginal}(\mathscr C) = \left\{
\mathbf q : \begin{array}{ll} p_C \ge 0 & \forall C \in \mathscr C, \\
\sum_{C \in \mathscr C} p_C = 1, \\
\sum_{C : v \in C} p_C = q_v & \forall v \in V
\end{array} \right\}.$$

Our fairness metrics of interest are mainly Schur-concave functions. If it does not hold that $\sum_{i=1}^n x_i = \sum_{i=1}^n y_i$ for some $\vec{x},\vec{y} \in \mathscr{C}$, then we compare $\vec{x}$ and $\vec{y}$ using weak majorization, focusing on functions that are concave, increasing, and symmetric in their inputs. These functions favor equality in the sense of the Pigou-Dalton transfer principle, where a transfer of utility from a richer to a poorer individual increases fairness. 
We also emphasize that unless the set of acceptable packings all have the same fixed cardinality, it is impossible to achieve both the maximum cardinality and a strictly concave fairness objective. Nonetheless, we prove that there is a bound on the maximum \emph{coverage loss}, the decrease in cardinality needed to ensure any vertex (included in some structure, for example, in a 2- or 3-cycle) is covered with strictly positive probability in the lottery. In the next section, we prove a bound on this coverage loss, along with providing background on Schur-convexity and a few functions of interest (utilitarian,  maximin and leximin, Nash social welfare, and the Gini coefficient).

\subsection{Fairness Metrics}

\label{sec:fairness}
We give several definitions to clarify the fairness metrics of interest in this paper.
Consider two vectors $\vec{x}, \vec{y} \in \mathbb{R}^n$. We say that $\vec{x}$ \emph{weakly majorizes} (or dominates) $\vec{y}$, denoted $\vec{x} \succ_w \vec{y}$, if for all $k = 1, 2, \ldots, n$,
$$\sum_{i=1}^k x_{(i)} \geq \sum_{i=1}^k y_{(i)},$$
where $x_{(1)} \leq x_{(2)} \leq \ldots \leq x_{(n)}$ and $y_{(1)} \leq y_{(2)} \leq \ldots \leq y_{(n)}$ are the components of $\vec{x}$ and $\vec{y}$ sorted in nondecreasing order. If $\vec{x},\vec{y}$ also satisfy $\sum_{i=1}^n x_i = \sum_{i=1}^n y_i$, then we say that $\vec{x}$ \emph{majorizes} (or dominates) $\vec{y}$, denoted $\vec{x} \succ \vec{y}$. If at least one inequality is strict, then $\vec{x}$ \emph{Lorenz dominates} $\vec{y}$, and a \emph{Lorenz dominant} vector is not Lorenz dominated by any other vector. Say that a function $f: \mathbb{R}^n \rightarrow \mathbb{R}$ is \emph{Schur-concave} if $\vec{x} \succ \vec{y}$ implies $f(\vec{x}) \geq f(\vec{y})$. Note that we flip many of these definitions from the usual convention, as we focus on the vectors sorted in nondecreasing order.

\paragraph{Utilitarian fairness.} The utilitarian fairness model maximizes the sum of the probabilities that each vertex is matched, $q_v$. More formally, we define it as
\begin{equation}
    F^{\text{UT}}(\vec{p}, \mathscr{C}) = \sum_{v \in V} q_v.
\end{equation}
To optimize for this objective, we require that each $C_i \in \mathscr{C}$ is a maximum-cardinality packing. In the social choice literature, one common desiratum is that packings be \textit{Pareto-efficient}, meaning maximal. Thus our requirement parallels this condition since all maximum packings are maximal. 

\paragraph{Maximin and leximin fairness.} The maximin fairness model maximizes the minimum probability that any vertex is matched. It is defined as:
\begin{equation}
    F^{\text{MM}}(\vec{p}, \mathscr{C}) = \min_{v \in V} q_v.
\end{equation}
Note that this objective is still linear in the context of \eqref{lp:fairness}. If we define the maximin probability $\lambda$ by adding the constraint
$$\lambda \leq q_v \quad \forall v \in V,$$
then we can rewrite the objective as $\max \lambda$.
The leximin fairness model is based on the maximin fairness model. An optimal leximin fair packing first maximizes the minimum probability that any vertex is matched. Out of those maximin fair packings, it maximizes the minimum probability of the next worst-off set of vertices, and so on. We can stretch our initial definition of a fairness metric by allowing the codomain to be a vector, $F: [0,1]^k \times \mathscr{C} \rightarrow [0,1]^n$. Then we can define the leximin fairness model as:
\begin{equation}
    F^{\text{LM}}(\vec{p}, \mathscr{C}) = \text{sorted}((q_v)_{v \in V}),
\end{equation}
the vector of packing probabilities sorted in nondecreasing order. The optimization problem \eqref{lp:fairness} is a lexicographic multiobjective one. It is clear that any optimal solution for $F^{\text{LM}}$ is also optimal for $F^{\text{MM}}$, as the first element of the sorted vector is the maximin probability. Thus we focus on just the leximin objective.

\paragraph{Nash social welfare.} Nash social welfare is the product of the probabilities that each vertex is matched. That is, 
\begin{equation}
    F^{\text{NW}}(\vec{p}, \mathscr{C}) = \prod_{v \in V} q_v.
\end{equation}
It can also be expressed as the sum of logs, $\sum_{v \in V} \log q_v$.
Nash social welfare lies somewhere between the utilitarian and maximin or leximin fairness objectives: the utilitarian objective cares equally about all patient-donor pairs, whereas the maximin and leximin objectives care about the worst-off patient-donor pair(s). The Nash social welfare objective balances the weight of the common good outcome with the outcomes of the worst-off individuals. One drawback of this objective is that if any vertex has $q_v = 0$, then the entire product is 0. 

\paragraph{Gini coefficient.} The Gini coefficient or index is a measure of statistical dispersion intended to estimate economic inequality. It is defined as
\begin{equation*}
    F^{\text{Gini}}(\vec{p}, \mathscr{C}) = \frac{\sum_{u,v \in V} |q_u - q_v|}{2n \sum_{v \in V} q_v}.
\end{equation*}
A Gini coefficient of 0 represents perfect equality, as all vertices have the same probability of being matched. On the other hand, a Gini coefficient of 1 represents maximal inequality among vertices. For consistency with our other fairness metrics (a higher score is better), we consider
\begin{equation}
    F^{\text{-GN}} = 1 - F^{\text{Gini}}
\end{equation}
as our fairness metric to be maximized. In contrast to the maximin and Nash social welfare objectives, the Gini coefficient can be nonzero even if some vertices have $q_v = 0$: it cares about the number of these zeros instead of whether one exists. Compared to the leximin solution, the Gini optimal solution will balance the number of zeros with the probabilities achieved by the other patient-donor pairs. However, unlike the Nash social welfare and leximin objectives, it does not necessarily prioritize the average number of vertices included in an exchange; a Gini optimal solution of all zeros has a Gini coefficient of 0.

\subsection{When do Fairness Functions Agree and Disagree?}
Having presented these fairness objectives, one might wish for a single lottery that optimizes them all. In that case, the practitioner does not need to carefully balance tradeoffs between different objectives. To that end, we first observe that for some sets of acceptable packings, optimizing for any fairness objective of interest will output the same lottery. In particular, if we restrict $\mathscr{C}$ to only contain packings of the same size, all concave fairness functions will agree.
\begin{lemma}
    \label{lem:leximin-nash}
    Suppose that for any acceptable packing, the sum of utilities $\sum_{v \in V} q_v$ is constant and maximal.
    Then an optimal solution for \eqref{lp:fairness} with respect to $F^{\text{LM}}$ is also optimal for any fairness function $F$ of interest for any acceptable packings $P: 2^V \rightarrow \mathbb{Z}_{\geq 0}$ whose rank (size of the packing) is submodular.
\end{lemma}
\proof{}
    In cooperative game theory, a game is concave if the characteristic function $v$ is submodular, i.e. for any $S, T \subseteq V$,
    $v(S) + v(T) \geq v(S \cup T) + v(S \cap T).$
    From a modification of \cite{Dutta_Ray_1989} (arranging the probabilities in nondecreasing order), the leximin solution found by iteratively maximizing the minimum $v(S)/|S|$ as in \autoref{alg:max-cardinality} is unique and Lorenz dominant in a concave game. For example, in nonbipartite graphs, for any $S \subseteq \Dtilde$ and its set of neighbors $N(S)$, take $v(S) = |N(S)| + \sum_{z \in S} (|z| - 1)$ as the maximum number of vertices covered in $S$. We can verify that $|N(S)|$ is submodular, as $N(S \cup T) = N(S) \cup N(T)$ and $N(S \cap T) \subseteq N(S) \cap N(T)$ for any $S, T \subseteq \Dtilde$. 
    Theorem~1 of \citeauthor{Bogomolnaia_Moulin_2004} is a special case of this result for bipartite graphs.
    Moreover, since the sum of utilities is constant, any Schur-concave function is maximized by the Lorenz dominant solution, including the Nash social welfare, negative of the Gini index (recall that a lower Gini index is more fair), and any symmetric or increasing concave function.
\endproof

Unfortunately, if the set of acceptable packings has a variable number of vertices covered, the fairness notions may differ, as we can see in the following examples.
In \autoref{fig:counterexample-23-cycle}, the Gini and leximin optimal solutions are to take the two-cycle with probability $\frac{1}{2}$ and the three-cycle with probability $\frac{1}{2}$. 
The utilitarian optimal solution is to take the three-cycle with probability $1$.
The optimal for the Nash social welfare lies in the middle, taking the two-cycle with probability $\frac{1}{3}$ and the three-cycle with probability $\frac{2}{3}$.
If we look at \autoref{fig:price-of-fairness-lb}, the leximin solution is to take the top and bottom two three-cycles each with probability $\frac{1}{2}$. The utilitarian and Gini optimal solutions now agree, always taking the bottom two three-cycles, and the Nash solution is to take the top three-cycle with probability $\frac{1}{5}$ and the bottom two three-cycles with probability $\frac{4}{5}$, illustrating that any agreement between different fairness functions is not consistent.

\begin{figure}[htbp!]
    \centering
    \begin{subfigure}{0.45\textwidth}
        \centering
        \begin{tikzpicture}[>=Stealth, scale=0.9, node distance=2cm]
    \node[circle,draw,fill=white] (A) at (0,0) {$v_1$};
    \node[circle,draw,fill=white] (B) at (0,2) {$v_2$};
    \node[circle,draw,fill=white] (C) at (3,2) {$v_3$};
    \node[circle,draw,fill=white] (D) at (3,0) {$v_4$};

    \draw[->,thick] (A) -- (B);
    \draw[->,thick] (B) -- (A);

    \draw[->,thick] (B) -- (C);
    \draw[->,thick] (C) -- (D);
    \draw[->,thick] (D) -- (B);
\end{tikzpicture}
    \caption{}
    \label{fig:counterexample-23-cycle}
    \end{subfigure}
    \begin{subfigure}{0.45\textwidth}
        \centering
        \begin{tikzpicture}[->,>=stealth,scale=0.8,vertex/.style={circle,draw,fill=white,inner sep=1pt}]
    \node[vertex] (v1) at (2.25,3.1) {$v_1$};
    \node[vertex] (v2) at (1.1,1.7) {$v_2$};
    \node[vertex] (v3) at (3.4,1.7) {$v_3$};
    \node[vertex] (v4) at (0.25,0) {$v_4$};
    \node[vertex] (v5) at (1.9,0) {$v_5$};
    \node[vertex] (v6) at (2.6,0) {$v_6$};
    \node[vertex] (v7) at (4.25,0) {$v_7$};

    \draw[thick,->] (v1) -- (v2);
    \draw[thick,->] (v2) -- (v3);
    \draw[thick,->] (v3) -- (v1);

    \draw[thick,->] (v2) -- (v4);
    \draw[thick,->] (v4) -- (v5);
    \draw[thick,->] (v5) -- (v2);

    \draw[thick,->] (v3) -- (v6);
    \draw[thick,->] (v6) -- (v7);
    \draw[thick,->] (v7) -- (v3);
\end{tikzpicture}
    \caption{}
    \label{fig:price-of-fairness-lb}
    \end{subfigure}
    \caption{Counterexamples illustrating the tension between utilitarian and fairness objectives.}
    \label{fig:counterexamples}
\end{figure}
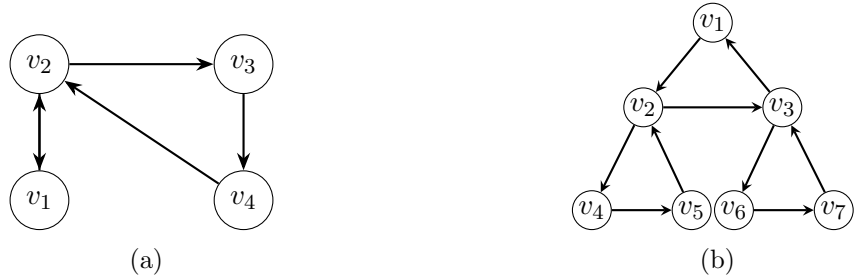

Moreover, \autoref{fig:counterexamples} highlights that it is impossible to achieve both the maximum utilitarian objective (in expectation) and the optimum of any other (strictly concave) notion of fairness unless we limit to the maximum cycle length to 2.
We can look at either example, say \autoref{fig:counterexample-23-cycle}.
Assuming that $\mathscr{C}$ only includes maximal exchanges, we have that 
$\mathscr{C} = \{C_1=(v_1,v_2), C_2=(v_2,v_3,v_4)\}$.
Then the set of achievable marginal probabilities takes the form $\vec{q} = (1-x, 1, x, x)$ for any $x \in [0,1]$.
The utilitarian objective is $2+x$, which is maximized at $x=1$ (i.e. always take the 3-cycle). However, since the other fairness objectives are strictly concave, they are maximized at some $x \in (0,1)$.

If we relax the requirement that we achieve the maximum cardinality in expectation, we can achieve better fairness, and there is a limit to how much we can lose in terms of the utilitarian objective.

\begin{proposition}
\label{prop:coverage-loss}
    Any vertex included in a 2-cycle or 3-cycle can be included in a packing of at most 3 less than the maximum cardinality. In general, if we allow up to $L$-cycles, any vertex included in one of those cycles can be covered with a loss of at most $(L-1)^2-1$.
\end{proposition}
\proof{}
    We first look at two- and three-cycles. Consider any maximum-cardinality packing $\mathcal{C}$. If $v$ is not covered, then either it is in a two-cycle with some vertex $u$, and $u$ is not free, or it is in a three-cycle with some vertices $u,w$, and at least one of $u,w$ is not free.
    Suppose $v$ is in a two-cycle. Then it can covered with a net loss of at most 1, since $u$ is matched with at most 2 other vertices in $\mathcal{C}$.
    Alternatively, suppose $v$ is in a three-cycle. In the worst case, both $u,w$ are matched with 2 other vertices in $\mathcal{C}$. Then we can cover $v$ by removing $u,w$ from their current covers, and matching $u$ with $v$ and $w$ with one of the two vertices it was previously matched with. This results in a net loss of at most 3.
    The lower bound of 3 is achieved by \autoref{fig:price-of-fairness-lb}, so the bound is tight.

    The case for general $L$ is similar. Consider any maximum-cardinality packing $\mathcal{C}$. If $v$ is not covered, then it must be in a $\ell$-cycle for some $2 \leq \ell \leq L$ with vertices $v_1,\ldots,v_{\ell-1}$, and at least one of $v_1,\ldots,v_{\ell-1}$ is not free.
    In the worst case, all $\ell-1$ other vertices are matched with $\ell-1$ other vertices in $\mathcal{C}$. Then we can cover $v$ by removing all $\ell-1$ other vertices from their current covers, and matching them with $v$ in the $\ell$-cycle. This results in a net loss of at most $(\ell-1)^2-1$, and $\ell=L$ gives the worst case.
\endproof

Having established the model, along with some intuition on the fairness objectives, we move on to an overview of results for solving the optimization problem \eqref{lp:fairness}.
\section{A Fully Polynomial-Time Approximation Scheme}
\label{sec:fptas}
Observe that solving $\eqref{lp:fairness}$ is at least as hard as finding a single maximum-weight acceptable packing in $\mathscr{C}$. For some packing structures, this problem is solvable with a polynomial-time algorithm. For other structures, this problem is NP-hard. 
In practice, that problem could still be efficiently solved using column generation, finding positive reduced-cost columns using the black box as we will do with our experiments in \autoref{sec:experiments-short}. However, in the worst case, we would call the black box $|\mathscr{C}|$ times ($O(2^n)$) to generate all columns. To bound the number of calls to this black box by a polynomial in $n$, we can employ the ellipsoid method, which is our main result framing the remainder of the paper. If the black box runs in polynomial time (as it does for matching), our algorithm immediately yields a polynomial-time algorithm, although we will give stronger guarantees for some of these cases in \autoref{sec:poly-algs-weights-paths}.

\begin{theorem}
    \label{thm:fptas}
    For any $\epsilon > 0$, we can find a $(1-\epsilon)$-approximation to \eqref{lp:fairness} in $\textsf{poly}(K_\text{sep},K_\text{eval},n,\frac{1}{\epsilon})$, where $K_\text{eval}$ is the time required to evaluate $F$ and $K_\text{sep}$ is the time required by the separation oracle.
\end{theorem}

To prove this theorem requires an approximation algorithm for determining whether $\vec{q} \in \text{marginal}(\mathcal{C})$, since $\mathscr{C}$ is exponentially large. 

\begin{lemma}
    \label{lem:approx-marginal-membership}
    There is an algorithm that finds an $\epsilon$-approximation to determine if $\vec{q} \in \textsf{marginal}(\mathscr{C})$ using $\textsf{poly}(n,1/\epsilon)$ calls to the oracle, where $n=|V|$.
\end{lemma}
\proof{}
    We relax the requirement that $\vec{q} \in \textsf{marginal}(\mathscr{C})$ exactly since $F$ is increasing, and instead allow $\vec{q} \in \hat{Q}$, where $\hat{Q} := \{\vec{q}: \exists \vec{p} \text{ satisfying } p_C \geq 0, \sum_{C \in \mathscr{C}} p_C \leq 1, \sum_{C: v \in C} p_C \geq q_v \forall v \in V \}$. This relaxation is valid because $F$ is monotone increasing and thus adding $\emptyset$ to $\mathscr{C}$ does not change the optimal value of $F(\vec{q})$. $\hat{Q}$ is a convex set consisting of both packing and covering constraints, and we can use the 
    fractional packing and covering algorithm of \cite{Plotkin_Shmoys_Tardos_1995} for approximately solving this problem.
\endproof

We can now give the proof of the theorem.
\proof[Proof of \autoref{thm:fptas}.]
    Consider the hypograph of $F$, 
    $\mathscr{S} := \{(\vec{q}^T,t)^T : \vec{q} \in \textsf{marginal}(\mathscr{C}), t \leq F(\vec{q})\},$
    where the maximum $t$ such that $\mathscr{S}$ is nonempty is the optimal value of \eqref{lp:fairness}.
    We can use the ellipsoid method to optimize, provided we have a separation oracle for $\mathscr{S}$ which verifies if $\vec{q} \in \textsf{marginal}(\mathscr{C})$ and $t \leq F(\vec{q})$, or returns a separating hyperplane. 
    The former is done via \autoref{lem:approx-marginal-membership}, and the latter can be checked by computing a supergradient of $F(\vec{q})$. That is, suppose $(\vec{r}^T, t')^T \notin \mathscr{S}$, and $\pi$ is a supergradient of $F$ satisfying $F(\vec{r}) + \pi \cdot (\vec{q}-\vec{r}) \geq F(\vec{q})$ for every $\vec{q}$. We can compute $\pi$ since $F$ is concave. Then $(\pi^T, -1)^T$ is a separating hyperplane, as $\pi^T \vec{r} - t' < \pi^T \vec{r} - F(\vec{r}) \leq \pi^T \vec{q} - F(\vec{q})$. 
    Thus the ellipsoid method solves this problem in $O(n^2 \log (1/\epsilon))$ iterations, each requiring one call to the separation oracle, one evaluation of $F$, and $O(n^2)$ additional work for updating the ellipsoid \citep{akgul1983topics,grotschel2012geometric}.
\endproof

In \autoref{sec:poly-algs-weights-paths}, we discuss the following polynomial-time algorithms  which can be used as a black box. 

\begin{restatable}{theorem}{matchingpoly}
    Given an unweighted undirected graph $G$, \eqref{lp:fairness} can be solved in polynomial time and yields the optimal solution for any fairness metrics of interest when the acceptable packings are maximum-cardinality matchings in $G$.
    \label{thm:max-cardinality}
\end{restatable}

\begin{theorem}
    Given an undirected weighted graph $G$, \eqref{lp:fairness} can be solved in polynomial time simultaneously for all fairness metrics of interest when the acceptable packings are maximum-weight maximum-cardinality matchings or more generally, maximum-weight fixed-cardinality matchings in $G$.
    \label{thm:max-weight}
\end{theorem}

\begin{theorem}
    \label{thm:2-paths-polytime}
    For $\ell \leq 2$, there is a polynomial time algorithm to find the maximum-cardinality packing of 2-cycles and $\ell$-paths starting from NDDs.
    However, for $\ell \geq 3$, finding a maximum-cardinality packing of $\ell$-paths starting from NDDs is NP-hard.
    \label{thm:3-paths-hard}
\end{theorem}

\autoref{thm:max-cardinality} was proved by \cite{Garcia-Soriano_Bonchi_2020,Li_Liu_Huang_Tang_2014}, but we give an alternate version of the proof in \autoref{sec:unweighted} which characterizes a more general set of acceptable packings that admit a Lorenz-dominant solution. In the remainder of
\autoref{sec:poly-algs-weights-paths}, we prove \autoref{thm:max-weight} for node weights and edge weights via reductions to the unweighted matching problem, and  \autoref{thm:2-paths-polytime} via a structural lemma analogous to Berge's and Hall's theorems in bipartite matching.
\section{Strongly Polynomial Algorithms for 2-cycles and 2-paths}
\label{sec:poly-algs-weights-paths}
If the set of acceptable packings are matchings (or more generally, independent sets of a matroid), then we can solve \eqref{lp:fairness} in strongly polynomial time for any concave increasing fairness metric $F$. We review known results for when the set of acceptable packings is the set of maximum-cardinality matchings, then extend those to maximum-weight maximum-cardinality matchings and maximum-weight matchings of fixed cardinality. Finally, we consider adding short paths. The proof of \autoref{thm:max-cardinality} was proved in \cite{Li_Liu_Huang_Tang_2014,Garcia-Soriano_Bonchi_2020}, albeit with different techniques. The proofs all crucially depend on the Edmonds-Gallai structure theorem, which characterizes the set of all maximum-cardinality matchings in nonbipartite graphs. 

\begin{proposition}[Edmonds-Gallai decomposition; see \cite{Lovasz_Plummer_1986}]
    \label{thm:gallai-edmonds}
        If $G=(V,E)$ is a graph, then $V$ can be decomposed into a set of deficient, adjacent, and critical vertices as follows. Let $D(G)$ be the set of all vertices not covered by at least one maximum matching. Let $A(G)$ be the set of vertices in $V(G) - D(G)$ adjacent to at least one point in $D(G)$. Finally let $C(G) = V(G) \setminus A(G) \setminus D(G)$. Then 
        \begin{enumerate}
            \item the components (blossoms) of the subgraph induced by $D(G)$ are factor-critical,
            \item if $M$ is any maximum matching of $G$, it contains a near-perfect matching of each component of $D(G)$, a perfect matching of each component of $C(G)$ and matches all vertices of $A(G)$ with vertices in distinct components of $D(G)$, and
            \item the matching number is $\nu(G) = \frac12 \left( |V(G)| - c(D(G)) + |A(G)| \right)$, where $c(D(G))$ denotes the number of components of the graph spanned by $D(G)$.
        \end{enumerate}
\end{proposition}

In the remainder of the paper, we use the notation $D(G), A(G), \text{ and } C(G)$ to refer to subsets of vertices defined by the decomposition, and $\tilde{D}(G)$ for the factor-critical components of $D(G)$.

\subsection{Maximum-Cardinality Matchings}
\label{sec:unweighted}
To gain intuition into the problem, we first consider unweighted maximum-cardinality matchings (only packing two-cycles). We review results which parallel \cite{Roth_Sonmez_Utku_Unver_2005,Garcia-Soriano_Bonchi_2020,Li_Liu_Huang_Tang_2014} on the Lorenz dominance of the leximin solution in this regime and explain why it implies the following result.
\matchingpoly*

The idea is to construct a new graph based on the Edmonds-Gallai decomposition of the original graph, and repeatedly solve a circulation problem to determine and fix the maximin probabilities. First, recall that the Edmonds-Gallai decomposition characterizes the structure of all maximum-cardinality matchings in a graph, separating vertices that are matched in every maximum matching from those that are not in polynomial time (see \cite{Lovasz_Plummer_1986}).

In the remainder of this section, we assume that $C(G)$ is empty, as these vertices are matched in every maximum matching. Any set of maximum matchings we find can be augmented by adding a perfect matching within $C(G)$.

We claim that it suffices to find a fair lottery over acceptable packings restricted to nodes in $A(G)$ and $D(G)$.
Furthermore, that we can restrict our attention to the bipartite graph induced by $A(G)$ and the odd components of $D(G)$.
From the subgraph induced by $A(G)$ and $D(G)$, we construct a bipartite graph by removing all edges between vertices of $A(G)$, contracting all blossoms in $D(G)$ into pseudonodes, and removing any parallel edges. Call this graph $G'=(V',E')$, with vertex set $V'= A(G) \cup \Dtilde$ where $\Dtilde$ is the contracted components from $D(G)$, with size $s_z$ for all pseudonodes $z \in \Dtilde$. From \autoref{thm:gallai-edmonds}, we know vertices in $A(G)$ are matched in every maximum matching. Thus, in order to solve \eqref{lp:fairness}, we need only focus on covering the pseudonodes in $D(G)$. Given a fractional matching in $G'$, the following lemma shows we can find a set of maximum matchings in $G$ that achieve the same coverage probabilities.

\begin{lemma}
Given a graph $G$ and its Edmonds-Gallai decomposition $C(G) \cup A(G) \cup D(G)$, we can find a fair lottery over acceptable matchings in $G$--that is, a set of $\mathscr{C}, \vec{p}$--by finding a stochastic matrix $P = [p_{uz}]_{u \in A(G), z \in \Dtilde},$ where $p_{uz} = \sum_{C_i \in \mathscr{C}} \mathbbm{1}((u,z) \in C_i) \cdot p_i$ describe the probability that edge $(u,z)$ is included in the matching. The matrix $P$ is feasible if a set of $\mathscr{C}, \vec{p}$ exist such that $\vec{p}$ is a valid probability distribution and $\mathscr{C}$ consists of $k$ acceptable matchings. Furthermore, we claim we can find $\mathscr{C}, \vec{p}$ from $P$ in polynomial time.
\label{lem:stochastic-matrix}
\end{lemma}
\proof{}
    The proof follows from a generalization of the Birkhoff-von Neumann theorem, specifically the construction of Lawler \& Labetoulle as presented in Theorem 11.9 of \citeauthor{Lenstra_Shmoys} on minimizing makespan for open shop scheduling. They show that you can always find a decrementing set $S \subseteq P$ with exactly one element of $P$ in each tight row and column, and at most one element of $S$ in each slack row and column. Then they use $S$ to find a partial schedule of length $\delta > 0$ (chosen to keep the tight rows and columns tight), subtract $\delta S$ from $P$, and repeat until $P$ is all zeros, which occurs in a polynomial number of iterations. 
    
    Think of the vertices of $A(G)$ as machines and the pseudonodes of $\Dtilde$ as jobs, and the probability $p_{uz}$ as the processing time that job $u \in A(G)$ requires on machine $z \in \Dtilde$. 
    Each row satisfies $\sum_z p_{uz} = 1$, and each column satisfies $\sum_u p_{uz} < 1$, which means there is an optimal solution with makespan 1 that can be found in polynomial time, and the set of partial schedules is a set of matchings with probabilities determined by their lengths. Since tight rows stay tight, each machine is fully scheduled, so each vertex in $A(G)$ is included in each matching.
\endproof

It remains to argue that we can find such a stochastic matrix $P$ for any desired fairness metric. We focus first on the objective functions $F^{\text{LM}}$ and $F^{\text{UT}}$, and later show that our solution is also optimal for $F^{\text{NW}}$ and any other Schur-concave function if there is a Lorenz dominant solution.
We introduce a new graph based on \autoref{thm:gallai-edmonds} and define a circulation problem parametrized by $\lambda$ on it, which will turn out to be the maximin probability of the optimal solution.
Next, we determine the value of $\lambda$ along with the desired distribution on the components of $D(G)$ by studying the necessary conditions for feasibility in the circulation problem.

\paragraph{Circulation.} We construct a circulation parametrized by $\lambda$ by introducing a source vertex $s$ and sink vertex $t$ to the bipartite graph $G'=(A(G) \cup \Dtilde, E')$. The source $s$ has an arc pointing to every vertex in $A(G)$ with capacity 1. All arcs in $E'$ are directed from $A(G)$ to $\Dtilde$ with infinite capacity. There is an arc from every vertex in $z \in \Dtilde$ to the sink $t$ with infinite capacity and a lower bound of $\max\{0, s_z \lambda - s_z+1\}$ on the flow. Finally, we add an arc of infinite capacity from $t$ to $s$ to complete the construction. An example is shown in \autoref{fig:flow-ex}.

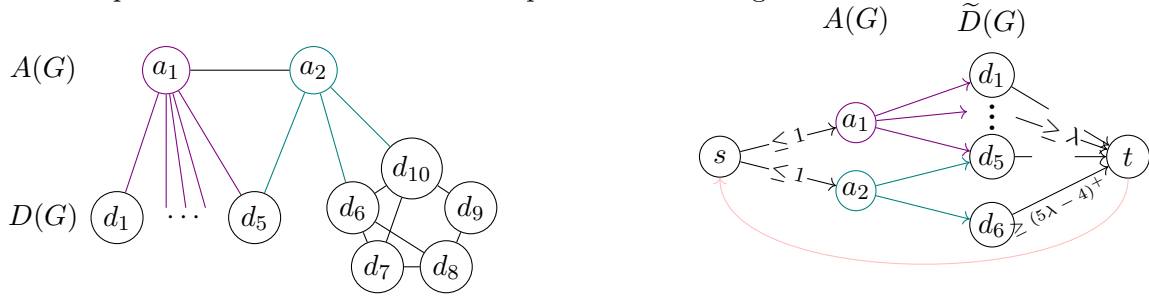
\begin{figure}[htbp!]
    \centering
    \vspace{-2em}
    \begin{subfigure}{0.45\textwidth}
        \centering
        \begin{tikzpicture}[scale=1.3, every node/.style={circle, draw, inner sep=2pt}]

    \node[draw=violet] (a1) at (1.5, 0) {$a_1$};
    \node[draw=teal] (a2) at (3, 0) {$a_2$};
    \node[draw=none] (A) at (0.25, 0) {$A(G)$};
    
    \node (d1) at (1, -1.5) {$d_1$};
    \node[draw=none] (dots) at (1.7, -1.5) {$\cdots$};
    \node (d5) at (2.4, -1.5) {$d_5$};
    \node[draw=none] (D) at (0.25, -1.5) {$D(G)$};
    
    \node (d6) at (3.4, -1.4) {$d_6$};
    \node (d7) at (3.65, -2) {$d_7$};
    \node (d8) at (4.35, -2) {$d_8$};
    \node (d9) at (4.6, -1.4) {$d_9$};
    \node (d10) at (4, -1) {$d_{10}$};

    \draw (d6) -- (d7);
    \draw (d7) -- (d8);
    \draw (d8) -- (d9);
    \draw (d9) -- (d10);
    \draw (d10) -- (d6);
    \draw (d7) -- (d10);
    \draw (d6) -- (d8);

    \draw[color=violet] (a1) -- (d1);
    \draw[color=violet] (a1) -- (1.5, -1.4);
    \draw[color=violet] (a1) -- (1.7, -1.4);
    \draw[color=violet] (a1) -- (1.9, -1.4);
    \draw[color=violet] (a1) -- (d5);

    \draw[color=teal] (a2) -- (d5);
    \draw[color=teal] (a2) -- (d6);
    \draw[color=teal] (a2) -- (d10);


    \draw (a1) -- (a2);
\end{tikzpicture}
    \end{subfigure}
    \hfill
    \begin{subfigure}{0.45\textwidth}
        \centering
        \begin{tikzpicture}[scale=0.9, every node/.style={circle, draw, inner sep=1pt}, every edge/.style={>=Latex, shorten >= 4pt}]
    \node[inner sep=3pt] (s) at (0, 1) {$s$};
    \node[draw=violet] (a1) at (2, 1.5) {$a_1$};
    \node[draw=teal] (a2) at (2, 0.5) {$a_2$};

    \node (d1) at (4, 2.2) {$d_1$};
    \node[draw=none] (d234) at (4, 1.7) {\bf $\vdots$};
    \node (d5) at (4, 1) {$d_5$};
    \node (d6) at (4, 0) {$d_6$};
    \node[inner sep=3pt] (t) at (6, 1) {$t$};
    
    \node[draw=none, fill=none] at (2, 3) {$A(G)$};
    \node[draw=none, fill=none] at (4, 3) {$\Dtilde$};
    
    \foreach \i in {1, 2} {
        \draw[->] (s) -- (a\i) node[midway,draw=none,fill=white,inner sep=0pt, sloped] {\scriptsize $\leq 1$};
    }
    \foreach \i in {1,5} {
        \draw[->] (d\i) -- (t);
    }
    \draw[->] (d234) -- (t) node[midway,rectangle, minimum height=1cm,draw=none,fill=white,inner sep=0pt, sloped] {\footnotesize $\geq \lambda$};
    \draw[->] (d6) -- (t) node[pos=0.85, draw=none, below, sloped,inner sep=0pt] {\hspace{-3.5em} \tiny $\geq (5\lambda-4)^+$};

    \foreach \i in {1,234,5} {
        \draw[color=violet,->] (a1) -- (d\i);
    }
    \foreach \i in {5,6} {
        \draw[color=teal,->] (a2) -- (d\i);
    }


    \draw[->,pink] (t) .. controls +(0,-2) and +(0,-2) .. (s);
\end{tikzpicture}
    \end{subfigure}
    \vspace{-2em}
    \caption{Edmonds-Gallai decomposition of a graph $G$ and the corresponding construction of the circulation problem.}
    \label{fig:flow-ex}
\end{figure}

The next question is to determine when the circulation problem is feasible, which is answered by the following theorem.
\begin{proposition}[\citeauthor{Ahuja_Magnanti_Orlin_1993}]
    A circulation problem with nonnegative lower bounds on arc flows $l_{ij}$ and upper bounds $u_{ij}$ is feasible if and only if for every set $S$ of nodes
    $$\sum_{(i,j) \in (\bar{S},S)} l_{ij} \leq \sum_{(i,j) \in (S,\bar{S})} u_{ij},$$
    where $\bar{S}$ denotes the complement of $S$ in the vertex set.
    \label{thm:circulation}
\end{proposition}

Furthermore, this condition characterizes the maximal feasible $\lambda$. We use the notation $N(v)$ to denote the set of neighbors to $v \in V$, and in directed graphs $N^-(v) \text{ and }N^+(v)$ for the set of neighbors of incoming and outgoing edges of vertex $v$, respectively. 
\begin{lemma}
    \label{lem:circ-tight-inopt}
    $\lambda$ is maximal if and only if the following equation holds and is tight for at least one subset $B \subseteq \Dtilde$.
    \begin{equation}
    \label{eq:circulation-feasibility}
    \lambda \leq \frac{|N^-(B)| + \sum_{z \in B} (s_z - 1)}{\sum_{z \in B} s_z} \quad \forall B \subseteq \Dtilde.
\end{equation}
\end{lemma}
\proof{}
    Observe that the condition of \autoref{thm:circulation} is trivially satisfied for any subset $S$ where there are no lower bounds on the arc flows from $S$ to $\bar{S}$ or if there are any arcs $(a,d)$ where $a \in S$ and $d \in \bar{S}$. It directly follows from \autoref{thm:circulation} that \eqref{eq:circulation-feasibility} holds for $\lambda$ to be feasible, and furthermore \eqref{eq:circulation-feasibility} must be tight for at least one subset $B \subseteq \Dtilde$ by maximality of $\lambda$.
\endproof

Furthermore, we can efficiently determine the maximal $\lambda$ (without solving the circulation problem). We define a \emph{class} of vertices $S \subseteq \Dtilde$ such that each $v \in S$ is adjacent to the same subset of vertices in $A(G)$. That is, there exists $S_A \subseteq A(G)$ such that $N^-(v) = S_A$ for all $v \in S$. The number of classes is at most the number of odd components in $D(G)$, $|\tilde{D}(G)| \leq n$. From $S_A$, we find $S_{\tilde{D}} = \{z \in \Dtilde: N^-(z) \subseteq S_A\}$. We call this pair of vertices a \textit{block}, so that the set of all blocks is defined as
$$\mathscr{B}= \bigcup_{z \in \Dtilde} \big\{(S_A,S_{\tilde{D}}): S_A = N^-(z), S_{\tilde{D}} = N^+(S_A) \big\}$$
The subgraph induced by each block forms a bipartite graph. Notice that the classes form a laminar family, meaning that there are at most $n$ blocks, as there can be at most one block for each class. Thus it suffices to set 
\begin{equation}
    \label{eq:set-lambda}
    \lambda = \min_{\mathscr{B}} \frac{|S_A| + \sum_{z \in S_{\tilde{D}}} (s_z - 1)}{\sum_{z \in S_{\tilde{D}}} s_z}.
\end{equation}

\begin{algorithm}
    \caption{Sampling from a fair lottery over maximum-cardinality matchings}
    \label{alg:max-cardinality}
    \begin{algorithmic}[1]
        \Require{An undirected graph $G=(V,E)$}
        \Ensure{A maximum-cardinality matching drawn from the leximin optimal probability distribution $\vec{p}$ over maximum-cardinality matchings}
        \State Find the Edmonds-Gallai decomposition $C(G), A(G), D(G)$ of $G$.
        \State Construct the bipartite graph $G' = (A(G) \cup \Dtilde, E')$ as described above.
        \State $H \gets \emptyset$ \Comment{the partition of $G'$ into blocks}
        \While{$G'$ is not empty}
            \State Find the set of blocks $\mathscr{B}$ of $G'$. 
            \State Set $\lambda$ as in \eqref{eq:set-lambda} and let $B$ be the (union of) block(s) that attain(s) the minimum. Remove $B$ from $G'$, and add $B$ to $H$.
        \EndWhile
        \State Sample a matching. For each $u \in A(G)$, generate a uniform random number $r_u \in \{1, \ldots, \textsf{deg}_H(u)\}$, where $\textsf{deg}_H(u)$ is the degree of $u$ in $H$. Match $u$ to the $r_u$-th neighbor of $u$ in $H$. For each $z \in \Dtilde$ not matched by any $u \in A(G)$, generate another uniform random number $r_z \in \{1, \ldots, s_z\}$ and generate the perfect matching found when the $r_z$-th vertex of the factor critical component represented by $z$ is removed.
         \State Augment the matching in $\mathscr{C}$ with a perfect matching within $C(G)$.
    \end{algorithmic}
\end{algorithm}
We are now ready to prove the full theorem.
\proof[Proof of \autoref{thm:max-cardinality}.]
We describe the full algorithm in \autoref{alg:max-cardinality}. The algorithm runs in polynomial time, as the Edmonds-Gallai decomposition can be found in polynomial time by \autoref{thm:gallai-edmonds} and so can finding the matchings at the end, and the while loop for finding blocks takes $O(n^2)$ time. 
Instead of sampling from the distribution, we could also explicitly construct a stochastic matrix $P$, then use the algorithm from \autoref{lem:stochastic-matrix} to find a probability distribution $\vec{p}$ from $P$ over $\mathscr{C}$ with a polynomial size support. The algorithm describes how to sample from the distribution because it is more efficent, but $P$ could also be constructed explicitly as follows.
Initialize $P$ to be the all zeros matrix of size $|A(G)| \times |\Dtilde|$.
For each factor critical component $z \in B_{\Dtilde}$, there are $s_z$ matchings that cover all but one vertex in $z$. Enumerate the matchings and track the number of times each edge $(z_1,z_2)$ is used, then add that number divided by $s_z$ to $p_{z_1,z_2}$.
Run a flow on the bipartite graph with edges only among the blocks to determine the $p_{uz}$ values for $u \in A(G), z \in \Dtilde$. Add these values to $P$, where $z$ can be any vertex in the component represented by pseudonode $z$ which is adjacent to $u$.
Finally, augment each matching with a perfect matching within $C(G)$.
The correctness of \autoref{alg:max-cardinality} follows from the optimality of the maximin solution at each iteration of the while loop, according to \autoref{lem:circ-tight-inopt}.
Finally, we appeal to \autoref{lem:leximin-nash} to prove the leximin solution is optimal for all Schur-concave fairness functions.
\endproof

\subsection{Maximum-Weight Maximum-Cardinality 2-Cycles}
\label{subsec:weights} 
We first restrict our attention to only 2-cycles (equivalent to thinking about matching in an undirected graph) and consider two types of weights: node and edge weights, which may be used to represent factors such as urgency for patients, compatibility or transplant success probability between two patient-donor pairs, or to explicitly incentivize including hard-to-match patients. For node weights, there is a weight $w_v$ on each vertex $v \in V$, and for edge weights, there is a weight $w(u,z)$ for each edge $(u,z)$. The set of acceptable matchings are maximum-weight maximum-cardinality matchings. We will reduce the weighted setting back to the unweighted setting to prove a strongly polynomial algorithm exists for any fairness objective of interest.

For the edge-weighted setting, we define the bipartite graph $H$ as follows.
Given the original graph $G$, find the Edmonds-Gallai decomposition and remove edges between vertices of $A(G)$. Shrink all blossoms into pseudonodes $z \in \tilde{D}(G)$, with size $s_z$. Let the edge weights be $w(u,z)=s_z + \epsilon w(u,z)$ for all $(u,z) \in E$. WLOG only take the (any) edge inducing the highest weight of $w(u,z)$ if there are parallel edges between any pair of $u,z$. Make two copies of this graph, with the partition $L_1 \cup R_1$ in the first and $L_2 \cup R_2$ in the second graph. Then, add edges between corresponding $z_1 \in R_1$ and $z_2 \in R_2$ originating from the same pseudonode, with weight $2(s_z - 1 + \epsilon w_z)$. We claim there is a bijection between $G$ and $H$.
\begin{lemma}
    \label{lem:max-weight-reduction}
    The problem of finding a set of maximum-weight maximum-cardinality matchings in $G$ can be reduced to finding maximum-weight perfect matchings in the bipartite graph $H$.
\end{lemma}
\proof{}
Let $w_z$ be the maximum weight of the matching within $z$, and $w(u,z)$ be the maximum weight of the perfect matching between $u \in A(G)$ and the pseudonode $z$. Note that it is possible $w(u,z) < w_z$.

Compared to a maximum-weight maximum-cardinality matching $M$, a maximum-weight perfect matching in this new graph has weight $2(2|M|-|A(G)|+\epsilon w(M))$, where $2|M|-|A(G)|$ is the number of matched vertices in $D(G)$. We use $\epsilon$ to faciliate a bicriteria optimization, and create the two copies to allow for vertices of $D(G)$ to be unmatched. Furthermore, notice that this graph is bipartite.

\paragraph{Bijection between matchings.} We show that there is a bijection between solutions on the original graph $G$ and on the (WLOG) left side of the graph in $H$. Given a maximum-weight perfect matching $M$ on $H$ with weight $2c_1 + 2c_2 \epsilon$, we know there are $c_1$ matched vertices in $R_1$. Furthermore, each vertex in $L_1$ is matched with a node in $R_1$, so the number of vertices matched within the left copy of $H$ is $|A(G)|+c_1=|A(G)| + \sum_{z \in R_1} (s_z-1)$, which is exactly the maximum cardinality of a matching of $G$. Any vertices of $R_1$ that are matched with vertices of $R_2$, we interpret as being unmatched in $G$. Observing that $c_2$ is the weight of a maximum matching within the left copy of the graph, we see that a similar argument holds for the weights. To extend the matching to $G$, our last step is just to consider the matchings involving pseudonodes $z \in \Dtilde$. We choose either the maximum-weight matching within components $w_z$ if $z$ is matched to a node in the right copy of the graph, or the maximum-weight perfect matching $w(u,z)$ if $z$ is matched to a node in the left copy of the graph. Then the weight of the matching induced by the left copy of the graph will be exactly $$c_2 = \sum_{z \in Z_1} w(u,z) + \sum_{z \in Z_2} w_z,$$
where $Z_1 = \{z \in R_1: M(z)=u, u \in L_1\} \text{ and } Z_2 = \{R_1: M(z)=v, v \in R_2\}.$
The other direction of the reduction is even simpler, as we can forget about the matchings within contracted components of $\Dtilde$. We can also copy the matching from the left copy of the graph to the right copy, and match unmatched nodes in $R_1$ to the corresponding nodes in $R_2$. By the same argument as the other direction, the maximum-weight maximum-cardinality matching in $G$ must be maximum weight in $H$.
\endproof

This reduction characterizes deficient vertices in maximum-weight maximum-cardinality matchings as a subset of blossoms from maximum-cardinality matchings that are induced by the tight edges. The lemma allows us to prove the desired theorem for edge weights, and the case of node weights reduces to finding the maximum-weight independent set in a matroid, so we can run a greedy algorithm. 

\proof[Proof of \autoref{thm:max-weight} for maximum-weight maximum-cardinality matching.]
First, consider the case of node weights.
From \autoref{thm:gallai-edmonds}, we can define the matroid $(V,\mathcal{I}')$, where $\mathcal{I}' = \{U \subseteq D(G): \exists M \in \mathscr{C}: M \text{ covers } U\}$, and apply a greedy algorithm to find a maximum-weight maximum-cardinality matching: sort nodes of $D(G)$ by decreasing weight, $v_1,v_2,\ldots,v_n$ and return the independent set $I=\{v_i: r(U_i) > r(U_{i-1})\}$, where $U_i=\{v_1,\ldots,v_i\}$ for $i=1,\ldots,n$.
Next, we consider edge weights. By \autoref{lem:max-weight-reduction}, it suffices to study the bipartite graph $H$.
Complementary slackness tells us that a perfect matching in $H$ is optimal if and only if its edges are all admissible.
Thus we can run the Hungarian algorithm to find the set of admissible edges to find the structure of all maximum-weight perfect matchings, and run \autoref{alg:max-cardinality} using only admissible edges. 
\endproof

\subsection{Maximum-Weight Fixed-Cardinality 2-Cycles}
\label{sec:max-weight-mu}

Suppose we relax the assumption that our acceptable matchings must be maximum cardinality. Recall that the matching number $\nu(G)$ is the size of a maximum-cardinality matching, and any maximal matching has cardinality at least $\frac{1}{2}\nu(G)$. Then given a target cardinality $\mu \geq \frac{1}{2} \nu(G)$, we wish to characterize the set of all maximum-weight matchings of cardinality $\mu$, then use this information to determine $\mathscr{C}$ and $\vec{p}$.

Similar to \autoref{lem:max-weight-reduction} for maximum-cardinality matchings, we reduce the problem to minimum-weight perfect matching. We define a new graph, which we call $H$: starting with the original graph, we add $2(\nu(G)-\mu) + \textsf{def}(G)$ new nodes, where $\nu(G)$ is the matching number and $\textsf{def}(G)=|\Dtilde| - |A(G)|$ is the deficiency of $G$. Add an edge of weight $W$ between each new node and each vertex in the original graph, where $W$ is at least as large as the maximum weight of any edge in the original graph. Finally, we flip the signs of all edge weights to interpret them as costs. The proof is similar to the maximum-cardinality one, with the same ideas carrying through.

\proof[Proof of \autoref{thm:max-weight} for fixed cardinality.]
We claim that by running Edmonds' algorithm for nonbipartite minimum-cost perfect matching on $H$, we can characterize all maximum-weight fixed-cardinality matchings. Even though $H$ is no longer bipartite, we still have a polyhedral description of the set of perfect matchings, described by the following primal and dual linear programs. 

\vspace{-5mm}
\begin{multicols}{2}
\begin{equation*}
    \begin{aligned}
        \min \quad & \sum_{e \in E} w_e x_e \\
        \text{s.t.} \quad & \sum_{e \in \delta(v)} x_{e} = 1 \quad \forall v \in V \\
        & \sum_{e \in \delta(U)} x_{e} \geq 1 \quad \forall U \subseteq V, |U| \geq 3 \text{ odd} \\
        & x_{e} \geq 0 \quad \forall e \in E 
    \end{aligned}
    \label{lp:min-weight-nonbipartite-matching}
    \end{equation*}
    \begin{equation*}
    \begin{aligned}
        \max \quad & \sum_{U \subseteq V: |U| \text{ odd}} y_U \\
        \text{s.t.} \quad & \sum_{U \subseteq V :e \in \delta(U)} y_U \leq w_e \quad \forall e \in E \\
        & y_U \geq 0 \quad \forall U \subseteq V, |U| \geq 3 \text{ odd}
    \end{aligned}
    \end{equation*}
\end{multicols}

There are decision variables $x_e$ corresponding to edge $e \in E$ being in the matching in the primal, and dual variables $y_U$ for $U \subseteq V$. We use the notation $\delta(U)$ to denote edges which are in the cut of vertex set $U$.

Complementary slackness requires that if $x_e > 0$, then $\sum_{U \subseteq V: e \in \delta(U)} y_U = w_e$. Thus any edges in the matching must be admissible edges. From there, similar arguments as in the bipartite setting for the maximum-cardinality case hold to (1) prove the bijection between $H$ and the original graph, (2) identify maximum-weight fixed-cardinality matchings after running Edmonds' algorithm for nonbipartite minimum cost matching on $H$, and finally (3) reduce the problem to the unweighted case and apply the results of \autoref{sec:unweighted}.
\endproof

\subsection{Paths}
\label{sec:paths}
We now consider allowing paths of length at most $\ell$ starting from a set of altruistic donors $A$, in addition to 2-cycles. 
We first introduce some notation.
Given a directed graph $G=(V \cup A, E)$, where $V$ is the set of patient-donor pairs, let $N(A)$ and $N^2(A)$ be the set of vertices in $V$ which are reachable by a path exactly 1 and 2 arcs away from a vertex in $A$ respectively; let $G'$ be the undirected graph with vertices $V$ and edges the 2-cycles of $G$.
Our main result is the following structural theorem, which then implies \autoref{thm:2-paths-polytime}.

\begin{lemma}
\label{thm:2-paths-structural}
    The following are equivalent:
    (a) $\mathcal{P}$ is maximum. (b) $\mathcal{P}$ admits no augmenting configurations. (c) There is a set $S \subseteq A$ such that $\textsf{def}(S) = \textsf{exp}(\mathcal{P})$, where $\textsf{def}(S) = |S|-|M(S)|$ and $\textsf{exp}(\mathcal{P})$ is the number of exposed NDDs with respect to $\mathcal{P}$.
\end{lemma}

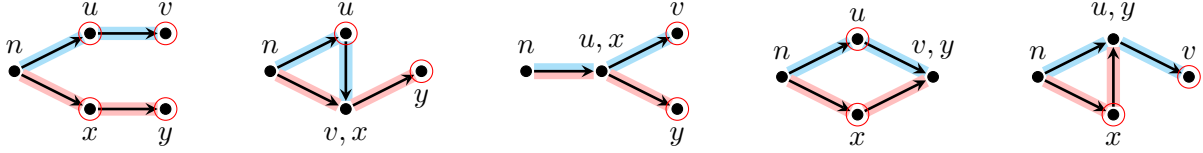
\begin{figure}[htbp!]
\centering
\tikzset{
    vertex/.style={circle, draw, fill=black, minimum size=4pt, inner sep=0pt, outer sep=1pt},
    circ/.style={circle, draw, red, minimum size=8pt, inner sep=0pt, outer sep=1pt},
    normal/.style={black, line width=1pt, ->},
    pinkedge/.style={line width=5pt, pink, postaction={draw, normal}},
    blueedge/.style={line width=5pt, cyan!30, postaction={draw, normal}},
    >=stealth
}
%
\begin{subfigure}{0.18\textwidth}
\begin{tikzpicture}[scale=1]
\node[vertex,label=above:$n$] (n2) at (-1,-0.5) {};
\node[vertex,label=above:$u$] (u2) at (0,0) {};
\node[vertex,label=above:$v$] (v2) at (1,0) {};
\node[vertex,label=below:$x$] (x2) at (0,-1) {};
\node[vertex,label=below:$y$] (y2) at (1,-1) {};

\draw[blueedge] (n2) -- (u2);
\draw[blueedge] (u2) -- (v2);
\draw[pinkedge] (n2) -- (x2);
\draw[pinkedge]  (x2) -- (y2);

\node[circ, fit=(u2)] (u2c) {};
\node[circ, fit=(v2)] (v2c) {};
\node[circ, fit=(x2)] (x2c) {};
\node[circ, fit=(y2)] (y2c) {};

\end{tikzpicture}
\end{subfigure}
\hfill
%
\begin{subfigure}{0.18\textwidth}
\begin{tikzpicture}[scale=1]

\node[vertex, label=above:$n$] (n5) at (-1,-0.5) {};
\node[vertex,label=above:$u$] (u5) at (0,0) {};
\node[vertex,label=below:{$v,x$}] (vx5) at (0,-1) {};
\node[vertex,label=below:$y$] (y5) at (1,-0.5) {};

\draw[blueedge] (n5) -- (u5);
\draw[blueedge] (u5) -- (vx5);
\draw[pinkedge] (n5) -- (vx5);
\draw[pinkedge] (vx5) -- (y5);

\node[circ, fit=(u5)] (u5c) {};
\node[circ, fit=(y5)] (y5c) {};

\end{tikzpicture}
\end{subfigure}
\hfill
%
%
\begin{subfigure}{0.18\textwidth}
\begin{tikzpicture}[scale=1]

\node[vertex, label=above:$n$] (n1) at (-1,-0.5) {};
\node[vertex,label=above:{$u,x$}] (ux1) at (0,-0.5) {};
\node[vertex,label=above:$v$]   (v1)  at (1,0) {};
\node[vertex,label=below:$y$]   (y1)  at (1,-1) {};

\draw[cyan!30, line width=3pt, transform canvas={yshift=1.5pt}] (n1) -- (ux1);
\draw[pink, line width=3pt, transform canvas={yshift=-1.5pt}] (n1) -- (ux1);
\draw[normal] (n1) -- (ux1);
\draw[blueedge] (ux1) -- (v1);
\draw[pinkedge] (ux1) -- (y1);

\node[circ, fit=(v1)] (v1c) {};
\node[circ, fit=(y1)] (y1c) {};

\end{tikzpicture}
\end{subfigure}
\hfill
%
\begin{subfigure}{0.18\textwidth}
\begin{tikzpicture}[scale=1]

\node[vertex,label=above:$n$]  (n3) at (-1,-0.5) {};
\node[vertex,label=above:$u$]  (u3) at (0,0) {};
\node[vertex,label=above:{$v,y$}] (vy3) at (1,-0.5) {};
\node[vertex,label=below:$x$]   (x3)  at (0,-1) {};

\draw[blueedge] (n3) -- (u3);
\draw[blueedge] (u3) -- (vy3);
\draw[pinkedge] (n3) -- (x3);
\draw[pinkedge] (x3) -- (vy3);

\node[circ, fit=(u3)] (u3c) {};
\node[circ, fit=(x3)] (x3c) {};

\end{tikzpicture}
\end{subfigure}
\hfill
%
%
\begin{subfigure}{0.18\textwidth}
\begin{tikzpicture}[scale=1]

\node[vertex,label=above:$n$]  (n4) at (-1,-0.5) {};
\node[vertex,label=above:{$u,y$}] (uy4) at (0,0) {};
\node[vertex,label=below:$x$]   (x4) at (0,-1) {};
\node[vertex,label=above:$v$]   (v4) at (1,-0.5) {};

\draw[blueedge] (n4) -- (uy4);
\draw[blueedge] (uy4) -- (v4);
\draw[pinkedge] (n4) -- (x4);
\draw[pinkedge] (x4) -- (uy4);

\node[circ, fit=(x4)] (x4c) {};
\node[circ, fit=(v4)] (v4c) {};

\end{tikzpicture}
\end{subfigure}
\caption{Five local configurations of 2-paths starting from an NDD $n$. Nodes circled in red are exposed in one of the packings, and thus potential linking points in alternating trails.}
\label{fig:paths-cases}
\end{figure}

We define an alternating trail, and augmenting configurations for this problem. An alternating trail with respect to a packing $\mathcal{P}$ is a sequence of arcs which alternates between two arcs in $\mathcal{P}$ and two arcs not in $\mathcal{P}$. It starts at a covered vertex and ends at an exposed vertex. The two arcs are analogous to edges in a typical alternating path, and there are five cases which are depicted in \autoref{fig:paths-cases}. We define three augmenting configurations with respect to a packing $\mathcal{P}$. Each configuration starts with an NDD $a$ and arcs $(a,u)$ and $(a,v)$ to distinct vertices $u,v \in V$.
First, $u$ and $v$ are both exposed vertices.
Second, one of $u$ and $v$ is covered and the other is exposed, and the covered vertex starts an alternating trail.
Third, both $u$ and $v$ are covered vertices, and both vertices start an alternating trail.

To prove the result of this section, we apply an extension of Hall's theorem to hypergraphs. Say that a set $F$ of edges is pinned by another set $K$ of edges if every edge in $F$ has a nonempty intersection with some edge in $K$. 
We use the notation $M(S)$ to denote a maximu-cardinality matching of edges in $S$. For kidney exchange graphs, we overload notation so that for $S \subseteq A$, $M(S)$ is a maximum-cardinality matching formed by edges which are the second arc of a 2-path starting from an NDD within $S$.

\begin{proposition}[\cite{Aharoni_Haxell_2000}]
    Let $H=(V+A,E)$ be a bipartite hypergraph. Then $H$ admits an $A$-perfect matching iff there is an assignment of a matching $M(A_0)$ in $N_H(A_0)$ for every $A_0 \subseteq A$ such that pinning $M(A_0)$ requires at least $|A_0|$ disjoint edges from $\bigcup \{M(A_\ell): A_\ell \subseteq A_0\}$.
    \label{prop:aharoni-haxell}
\end{proposition}
From the proposition, the proof of the structural lemma and subsequently, of \autoref{thm:2-paths-polytime} follows.

\proof[Proof of \autoref{thm:2-paths-structural}.]
    To prove (a) $\Leftrightarrow$ (c), we can apply \autoref{prop:aharoni-haxell}. Observe that $\bigcup \{M(A_l): A_l \subseteq A_0\}$ for any subset of NDDs $A_0 \subseteq A$ is simply $M(A_0)$ again. To pin $M(A_0)$, we need all of $M(A_0)$ since matchings are vertex disjoint, and $|A_0| \leq |M(A_0)|$, which implies there is an $A$-perfect matching iff $|A_0| \leq |M(A_0)|$ for all $A_0 \subseteq A$. Reformulating this Hall-type condition in terms of the deficiency gives the desired equivalence.

    It is straightforward to see that (a) $\Rightarrow$ (b), as each augmenting configuration increases the cardinality of a packing by 2. We show that (b) $\Rightarrow$ (c). 
    Let $\mathcal{P}$ be a packing with no augmenting configurations.
    Then let $S \subseteq A$ be the set of NDDs reachable from exposed NDDs via alternating trails of even length with respect to the 2-paths (so the total arc length is divisible by 4). If $X \subseteq S$ is the set of exposed NDDs of $S$, we know that $M(S) \geq |S|-|X|$, since each covered NDD in $S$ can be mapped to a disjoint 2-path which corresponds to an edge of $M(S)$. 
    Assume to the contrary that $|M(S)| > |S| - |X|$. Then there exists vertices $u,v$ covered by a maximum matching of $N(S) \cup N^2(S)$ which are exposed in $\mathcal{P}$. First, suppose $(u,v) \in M(S)$ form an edge not in the current packing $\mathcal{P}$. Since $(u,v) \in M(S)$, it starts from some NDD $a$ which must be in an even length alternating trail starting from an exposed NDD $a_e$. But then we can find an augmenting configuration starting from $a_e$: there is an alternating trail from $a_e$ to $a$, and we can add $(a,u,v)$ to make a type 1 augmenting configuration. Otherwise, even if $u,v$ do not form an edge, the same argument holds. They are in the neighborhood of some NDDs $a_u,a_v$ (distinct, to have $M(S) > |S|-|E|$), so we can match $a_u$ to $u$ and $a_v$ to $v$. If $a_u,a_v,a$ are all distinct, we have a type 3 augmenting configuration from $a_e$ to $a_u$ and $a_v$. If $a_u$ or $a_v$ is the same as $a$, we have type 2 instead. Since we assumed $\mathcal{P}$ had no augmenting configurations, we have a contradiction, so $|M(S)| = |S| - |X|$, which implies $\textsf{def}(S) = |S| - |M(S)| = |X| = \textsf{exp}(\mathcal{P})$.
\endproof

\proof[Proof of \autoref{thm:2-paths-polytime}.]
The case for $\ell=1$ is simply matching, adding edges for each 2-cycle and between all NDDs and their neighbors. The algorithm for $\ell=2$ is similar to the one proposed in \cite{Brewster_Hell_Pantel_Rizzi_Yeo_2003} for packing paths of length 1 and 2. We start with any packing $\mathcal{P}$ and maintain a set $S$ to eventually certify optimality. Start with $S$ as the set of exposed vertices with respect to $\mathcal{P}$. We build a search tree starting from $S$, alternating between edges not in $\mathcal{P}$ and edges in $\mathcal{P}$, to either discover an augmenting configuration, or we increase the set $S$. If we can no longer increase $S$ and have not found an augmenting configuration, then by \autoref{thm:2-paths-structural}, $\mathcal{P}$ is maximum.

    \paragraph{Hardness of $\ell \geq 3$.} The reduction is from 3-dimensional matching, similar to \cite{Abraham_Blum_Sandholm_2007}.
    Suppose there is a 3DM instance which has a perfect matching. We can construct a covering of $H$ as follows: for every triple $(x_a,y_b,z_c)$ in the matching, add the paths $(s_x^i, 1, 2, \ldots, L-1, x_a)$, $(s_y^i, 1, 2, \ldots, L-1, y_b)$, $(s_z^i, 1, 2, \ldots, L-1, z_c)$, and $(s^i, x_a^i, y_b^i, z_c^i)$ to the covering. For every  triple $(x_a,y_b,z_c)$ not in the matching, add the paths $(s_x^i, 1, 2, \ldots, L-1, x_a^i)$, $(s_y^i, 1, 2, \ldots, L-1, y_b^i)$, and $(s_z^i, 1, 2, \ldots, L-1, z_c^i)$ to the covering, leaving $s^i, x_a, y_b, z_c$ exposed. It is easy to see that this is a valid covering including all non-NDD vertices.

    Conversely, suppose there is a covering of $H$ including all non-NDD vertices. Since the $x_a^i, y_b^i, z_c^i$ vertices in the gadget are unique to each triple, they must be covered by the path $(s^i, x_a^i, y_b^i, z_c^i)$, or by the paths $(s_x^i, 1, 2, \ldots, L-1, x_a^i)$, $(s_y^i, 1, 2, \ldots, L-1, y_b^i)$, and $(s_z^i, 1, 2, \ldots, L-1, z_c^i)$. 
    These cases determine if $x_a,y_b,z_c$ are covered or exposed, and thus if the triple is in the matching or not respectively. Since all $x_a,y_b,z_c$ must be covered by exactly one triple, we have a perfect matching.
\endproof

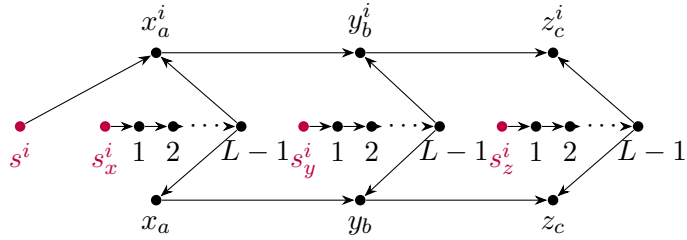
\begin{figure}[htbp!]
        \centering
        \begin{tikzpicture}[>=Stealth,scale=0.75]
    \tikzset{
  dot/.style = {circle, fill, inner sep=1.4pt},
  textnode/.style = {inner sep=2pt, font=\small},
  arrow/.style = {-Stealth},
}

\node[dot, label=below:{\color{purple}$s^i$},purple] (si) at (-1.5,0) {};

\node[dot,label=below:{\color{purple}$s_x^i$},purple] (sx) at (0,0) {};
\node[dot,label=below:{$1$}] (x1) at (0.6,0) {};
\node[dot,label=below:{$2$}] (x2) at (1.2,0) {};
\node[textnode] (xdots) at (1.8,0) {$\cdots$};
\node[dot] (xL) at (2.4,0) {};
\node at (2.65,-0.44){$L-1$};

\node[dot,label=above:{$x_a^i$}] (xai) at (0.9,1.3) {};
\node[dot,label=below:{$x_a$}] (xa) at (0.9,-1.3) {};

\draw[arrow] (sx) -- (x1);
\draw[arrow] (x1) -- (x2);
\draw[arrow] (x2) -- (xdots);
\draw[arrow] (xdots) -- (xL);
\draw[arrow] (xL) -- (xa);
\draw[arrow] (xL) -- (xai);

\node[dot,label=below:{\color{purple}$s_y^i$},purple] (sy) at (3.5,0) {};
\node[dot,label=below:{$1$}] (y1) at (4.1,0) {};
\node[dot,label=below:{$2$}] (y2) at (4.7,0) {};
\node[textnode] (ydots) at (5.3,0) {$\cdots$};
\node[dot] (yL) at (5.9,0) {};
\node at (6.15,-0.44){$L-1$};

\node[dot,label=above:{$y_b^i$}] (ybi) at (4.5,1.3) {};
\node[dot,label=below:{$y_b$}] (yb) at (4.5,-1.3) {};

\draw[arrow] (sy) -- (y1);
\draw[arrow] (y1) -- (y2);
\draw[arrow] (y2) -- (ydots);
\draw[arrow] (ydots) -- (yL);
\draw[arrow] (yL) -- (yb);
\draw[arrow] (yL) -- (ybi);

\node[dot,label=below:{\color{purple}$s_z^i$},purple] (sz) at (7,0) {};
\node[dot,label=below:{$1$}] (z1) at (7.6,0) {};
\node[dot,label=below:{$2$}] (z2) at (8.2,0) {};
\node[textnode] (zdots) at (8.8,0) {$\cdots$};
\node[dot] (zL) at (9.4,0) {};
\node at (9.65,-0.44){$L-1$};

\node[dot,label=above:{$z_c^i$}] (zci) at (7.9,1.3) {};
\node[dot,label=below:{$z_c$}] (zc) at (7.9,-1.3) {};

\draw[arrow] (sz) -- (z1);
\draw[arrow] (z1) -- (z2);
\draw[arrow] (z2) -- (zdots);
\draw[arrow] (zdots) -- (zL);
\draw[arrow] (zL) -- (zc);
\draw[arrow] (zL) -- (zci);

\draw[arrow] (xai) -- (ybi);
\draw[arrow] (ybi) -- (zci);
\draw[arrow] (xa) -- (yb);
\draw[arrow] (yb) -- (zc);
\draw[arrow] (si) -- (xai);

\end{tikzpicture}
        \caption{Gadget for reduction from 3DM to maximum covering with 2-cycles and paths of length 3. 
        }
        \label{fig:3dm-reduction}
    \end{figure}

Finally, we remark that we can combine the algorithm for packing 2-paths with the typical one for packing 2-cycles (i.e., matching). The high-level idea is to allow elements of alternating trails to be either two second-arcs as in \autoref{fig:paths-cases} or edges as in typical alternating paths.
We can assume WLOG that none of the second-arcs are also 2-cycles, or else we can just treat them as 2-cycle edges, since the NDDs do not contribute to the cardinality of a packing. We begin by extending the augmenting configurations to include either the elements in \autoref{fig:paths-cases} or the standard alternating paths from matching theory. To ensure these paths alternate in a valid way, we keep track of whether edges are from 2-cycles or from second-arcs of 2-paths. Then, when searching for augmenting configurations, the search tree must alternate between two 2-cycle edges or two pairs of second-arc edges. The main adjustment we need is on labeling vertices to find blossoms for the 2-paths. In the typical matching setting, we alternate between labeling vertices with ``even'' and ``odd,'' and finding an even-even edge means we have found a blossom. Instead, as we run BFS, we label both the second and third vertices (the patient-donor pairs) the opposite of the NDD in a 2-path, and only count an even-even edge as a blossom if it is not a second-arc in a 2-path. If we wish to include weights on arcs, we can directly apply Edmonds' algorithm for weighted nonbipartite matching with this modification for identifying blossoms. Then, the same reduction holds where we use the output of Edmonds' algorithm to find maximum-cardinality packings using only admissible edges.

\section{Computational Experiments}
\label{sec:experiments-short}

We also investigate the computational implications of our fairness models, using kidney exchange data from the APKD, but replacing the ellipsoid algorithm with column generation in the usual way. Our results demonstrate both the practicality of constructing optimally fair lotteries by column generation and the lifesaving impact of explicitly considering individual fairness when designing mechanisms. In \autoref{app:sec:static-pools}, we limit discussion to packings that include the maximum number of nodes with cycles of length at most 3 and unbounded paths, but later show that comparable results hold for several other packing structures and demonstrates the tradeoffs among the various packing structures in \autoref{app:subsec:packing-sets}. In \autoref{sec:dynamic-pools}, we examine the impact of randomization in dynamic pools.

A set of 25 KEP instances (which we will refer to as KEP1) was generated taking induced subgraphs from a cumulative historical compatibility graph obtained from the APKD \citep{Ashlagi_2025}. NDDs were added to these instances by randomly selecting pairs from which only donor properties and compatibilities were kept. Each generated instance consists of 200 patient-donor pairs and 10 NDDs. To demonstrate that analogous results hold on publicly available data, a second set of 25 synthetic KEP instances (KEP2) was generated using a procedure outlined in the appendix of \cite{Ashlagi_Roth_2021}, on which we conducted a limited set of experiments.
On a Lenovo laptop with a 3.20GHz AMD Ryzen 7 processor and 16GB of RAM, finding a lottery of maximum-cardinality packings of cycles of length at most 3 and unbounded paths optimizing the leximin objective on one of the KEP1 instances takes about ten seconds.

In our experiments, we consider three randomized algorithms to find packings of cycles and paths. The first, which we call
{\sf Implicit}, 
is randomized only due to the nondeterminism inherent in integer linear programming (ILP) solvers. {\sf Heuristic} has an additional step of randomizing the order of variables and constraints in the ILP formulation. The {\sf Leximin} algorithm draws a packing from the lottery that maximizes the leximin objective.
The ILP formulation used for {\sf Implicit} and {\sf Heuristic} as well as for column generation in {\sf Leximin} is due to \cite{Anderson_Ashlagi_Gamarnik_Roth_2015}. This formulation has been used by multiple kidney exchange programs, and can be modified to obtain packings that fall into one of many different classes (differentiated by allowed substructures and other properties).

\begin{figure}[htbp!]
\centering
\begin{minipage}{0.45\textwidth}
        \centering
        \includegraphics[width=\textwidth]{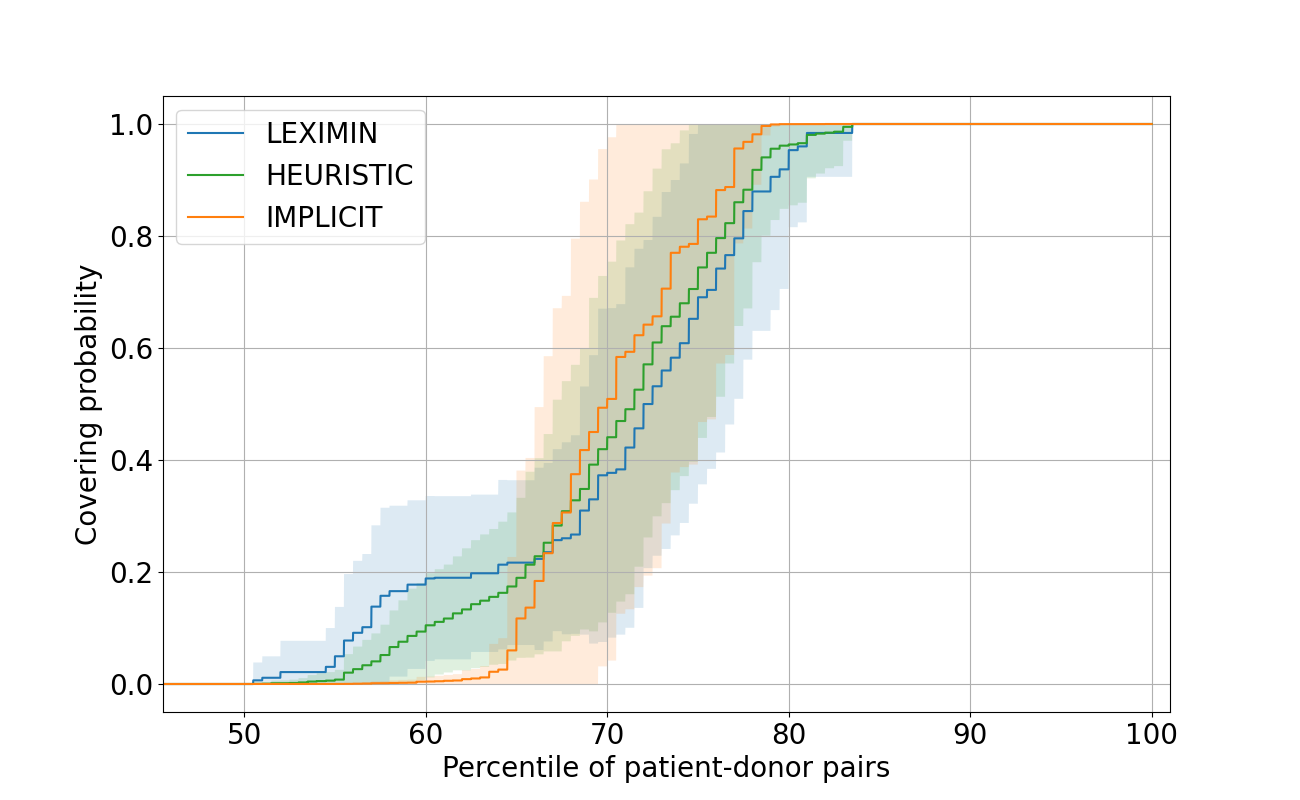}
    \end{minipage}
    \begin{minipage}{0.45\textwidth}
        \centering
        \includegraphics[width=\textwidth]{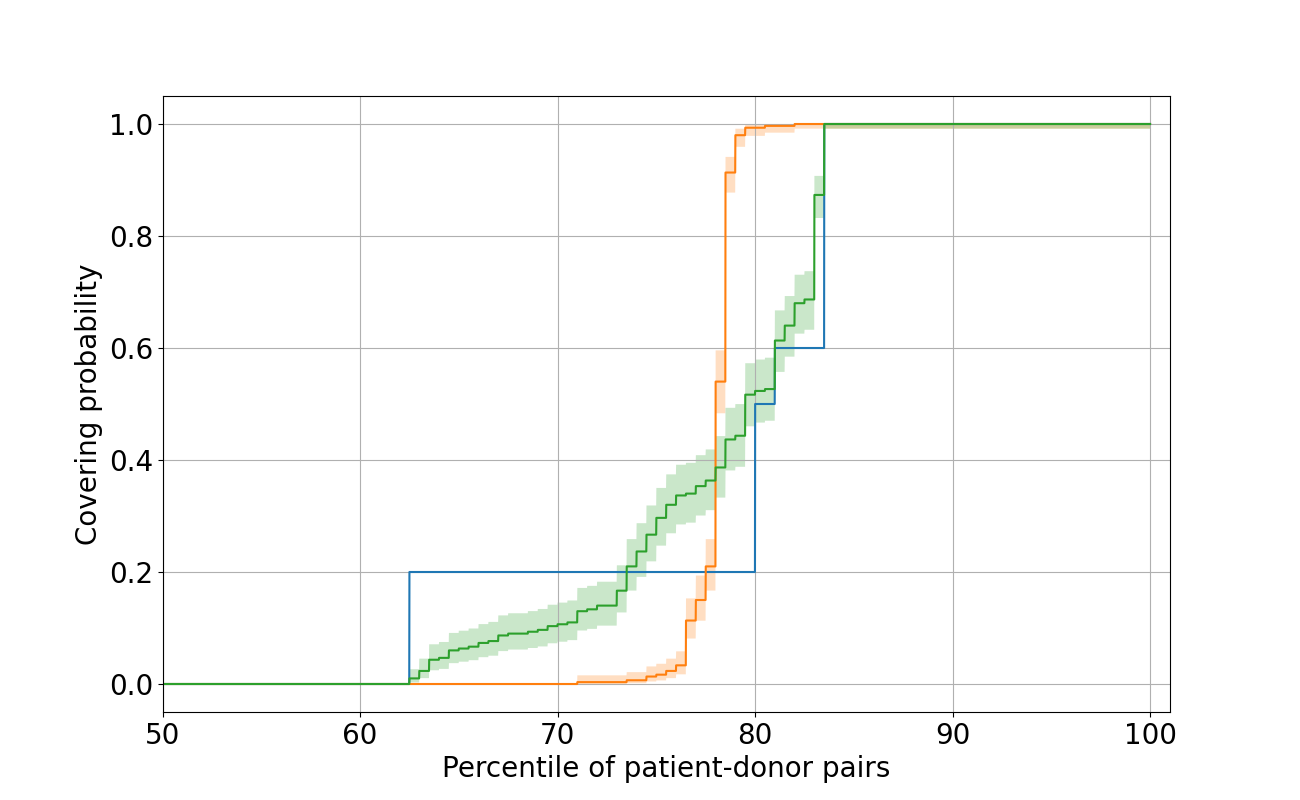}
    \end{minipage}
    \caption{Patient-donor pair inclusion probabilities for KEP1, with cycles of length at most 3 and unbounded paths. 
    The right plot shows results for one input, and shaded regions are 95\% confidence intervals. The left plot shows the averaged distribution over all 25 inputs, and shaded regions indicate one standard deviation.}
\label{fig:unbounded-paths}
\end{figure}

\subsection{Synthetic Instance Generation}

The 25 KEP2 instances were generated using APKD summary statistics from 2010 to 2019 given in \cite{Ashlagi_Roth_2021}. These include a distribution of patient-donor blood type pairings and conditional distributions of patient panel reactive antibodies (PRAs) by blood type pairing. For each patient-donor pair in the generated graph, a blood type pairing is drawn and then a patient PRA value is drawn from the conditional distribution. For NDDs, blood type is drawn from the global blood antigen distribution. Directed edges are then added based on blood type compatibility and probabilistic PRA tests, in which a random number from 0 to 100 is drawn and the test is negative if the number is greater than the prospective recipient's PRA value. These instances have 200 patient-donor pairs and 10 NDDs, the same as KEP1. However, the KEP2 instances are more densely connected, having 9213 directed edges on average compared to 2426 for KEP1. In future work, we will investigate the impact of calibrating (through a randomized selection) these inputs to the density of the inputs observed in KEP1.

\subsection{Static Pools}
\label{app:sec:static-pools}
 For each of the KEP1 instances, we compute the distribution of inclusion probabilities of patient-donor pairs for each of the three algorithms. For {\sf Leximin}, these probabilities are explicitly defined by the algorithm and are thus exact. For {\sf Implicit} and {\sf Heuristic}, they are defined implicitly and can be estimated by counting node inclusions over multiple runs.

Results can be seen in \autoref{fig:unbounded-paths}. {\sf Leximin}, due to its prioritization of the least likely pairs, provides a 1-in-5 chance to pairs that are practically guaranteed to be excluded by {\sf Implicit}. It is particularly noteworthy that the performance of {\sf Heuristic} is comparable to that of {\sf Leximin}, even though it is simpler to implement in practice. {\sf Leximin} still yields a significant redistribution of probability to those least likely to be matched, though: we see that 2.7\% of the coverage is shifted from above the crossover point to below it. Nonetheless, in \autoref{app:sec:heuristic-implementation} we further examine {\sf Heuristic} and consider the performance of the similar algorithm which shuffles node names in the KEP input (instead of shuffling ILP variable and constraints), as we'd expect this approximation to {\sf Heuristic} to be a more convenient alternative for practitioners to implement in practice. We find that there is not a significant difference between these two alternatives.

\subsection{Varying the Set of Acceptable Packings}
\label{app:subsec:packing-sets}
We document the results of running {\sf Leximin}, {\sf Heuristic}, and {\sf Implicit} on several classes of packings. Estimated coverage probabilities are computed over 300 runs. Reported confidence intervals are Jeffreys intervals.

\begin{figure}[htbp!]
\centering
\begin{minipage}{0.45\textwidth}
        \centering
        \includegraphics[width=0.9\textwidth]{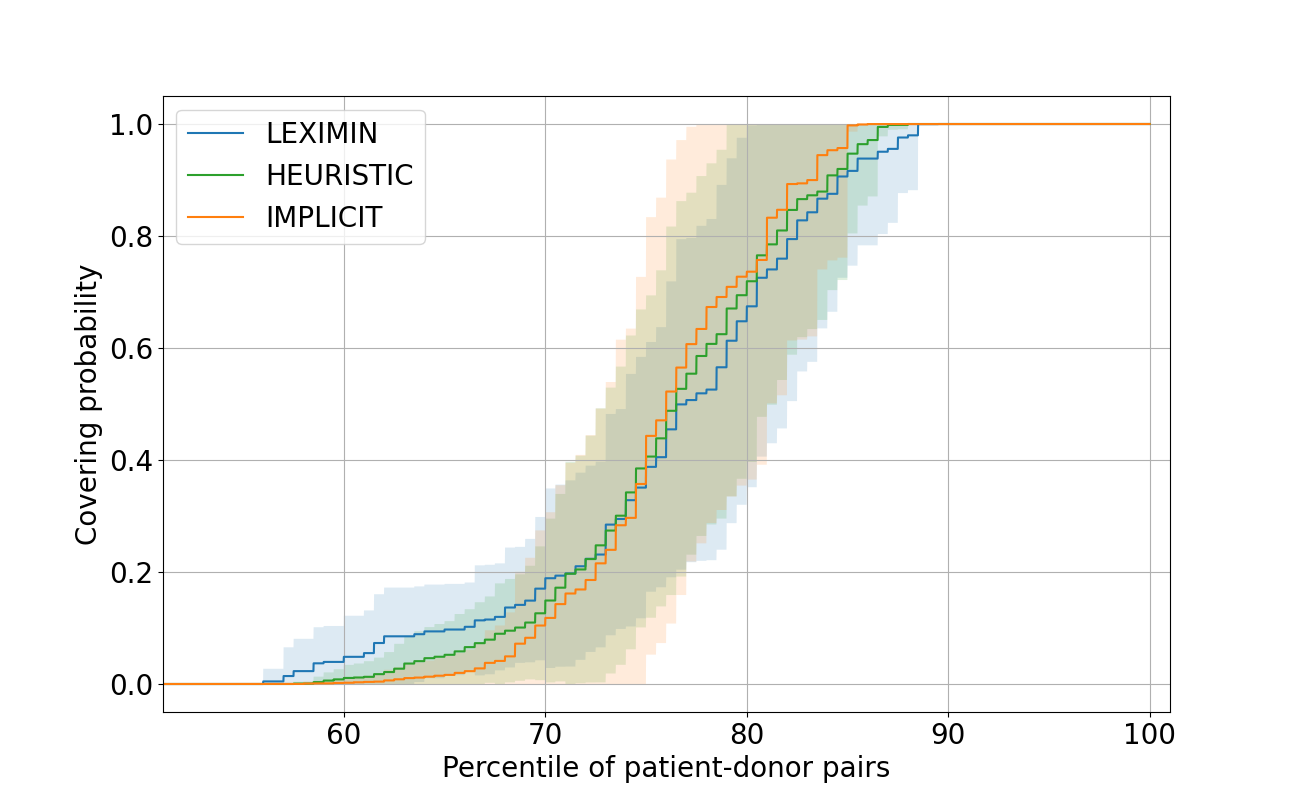}
    \end{minipage}
    \begin{minipage}{0.45\textwidth}
        \centering
        \includegraphics[width=0.9\textwidth]{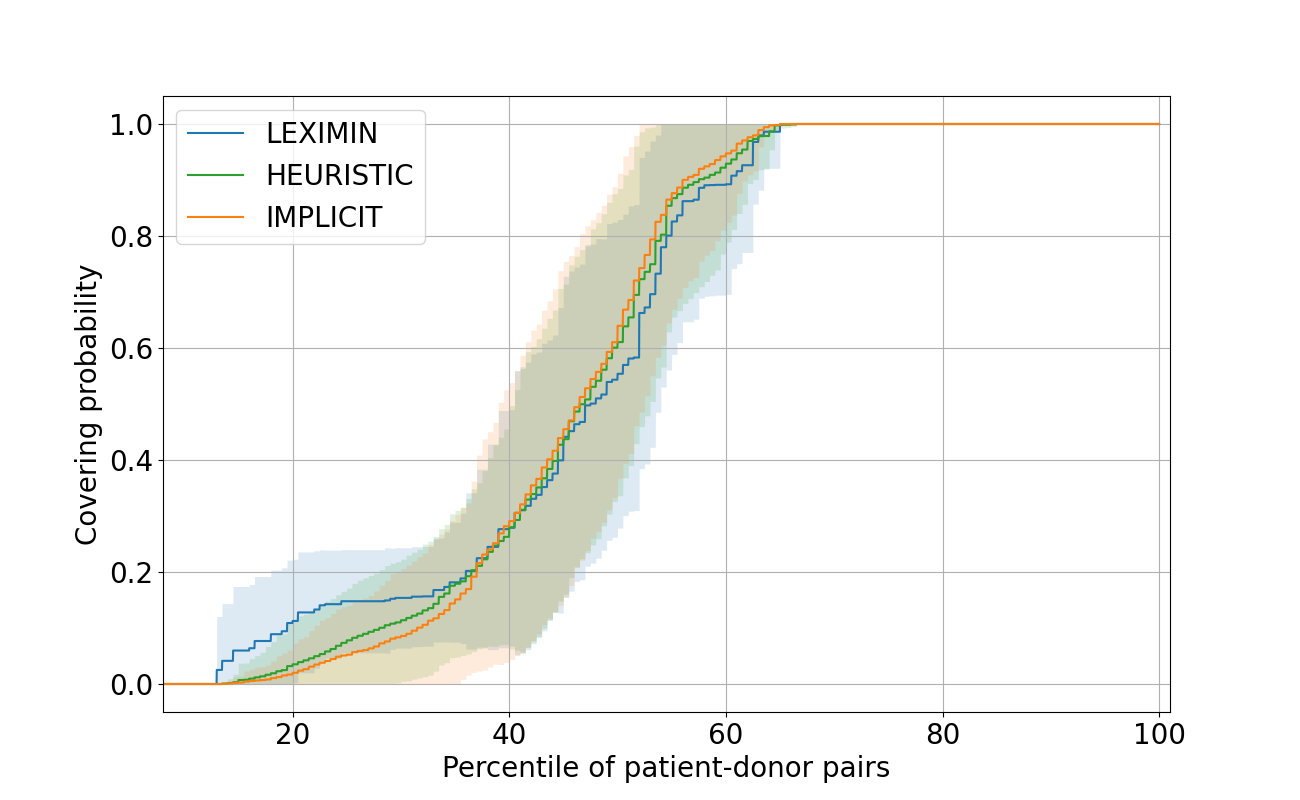}
    \end{minipage}
    \caption{Inclusion probabilities for maximum-cardinality 2- and 3-cycle packings. The plot on the left shows results for the KEP1 instances, with KEP2 results on the right.}
\label{fig:two-three}
\end{figure}

We first investigate the packing class of maximum-cardinality 2- and 3-cycle packings.
(Note that throughout, we will use the term ``maximum cardinality'' to refer to those packings for which the set of nodes included are of maximum cardinality, rather than the number of objects in the packing are of maximum cardinality.) The greater coverage seen in the KEP2 instances (\autoref{fig:two-three}) is a reflection of the stark difference in density between the two sets of instances, although there are still many nodes that are not covered at all in both cases. We can compute the number of nodes that are included in every packing by iteratively finding packings that minimize the number of nodes that have been included in every packing seen so far, using the same subroutine as used for column generation. These values for KEP1 are shown in \autoref{fig:two-three-hists}, which documents for each metric value, the number of the 25 inputs for which that value is obtained. 
Observe that with 3-cycles introduced, {\sf Leximin} may assign a probability of 1 to a node even if there exists a packing that does not include that node.

\begin{figure}[htbp!]
\centering
\begin{minipage}{0.3\textwidth}
        \centering
        \includegraphics[width=0.9\textwidth]{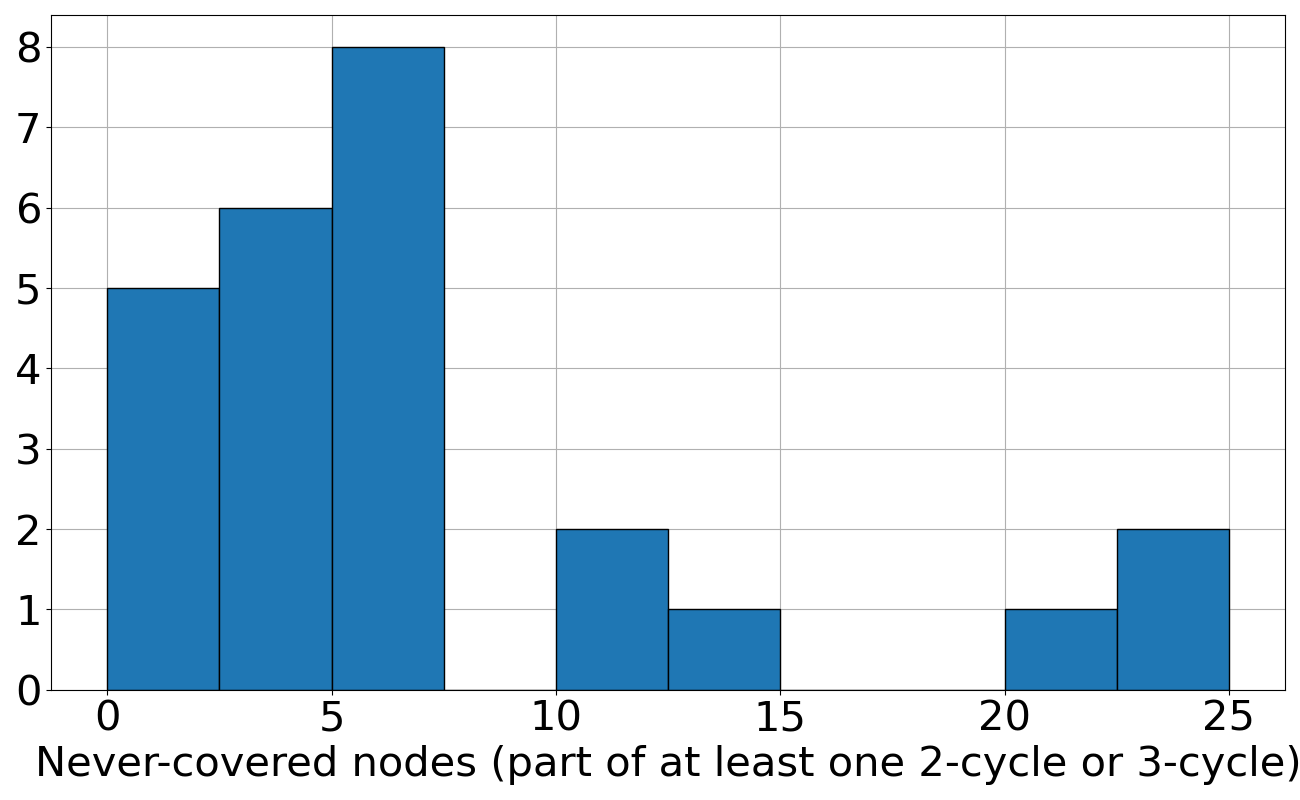}
    \end{minipage}
    \begin{minipage}{0.3\textwidth}
        \centering
        \includegraphics[width=0.9\textwidth]{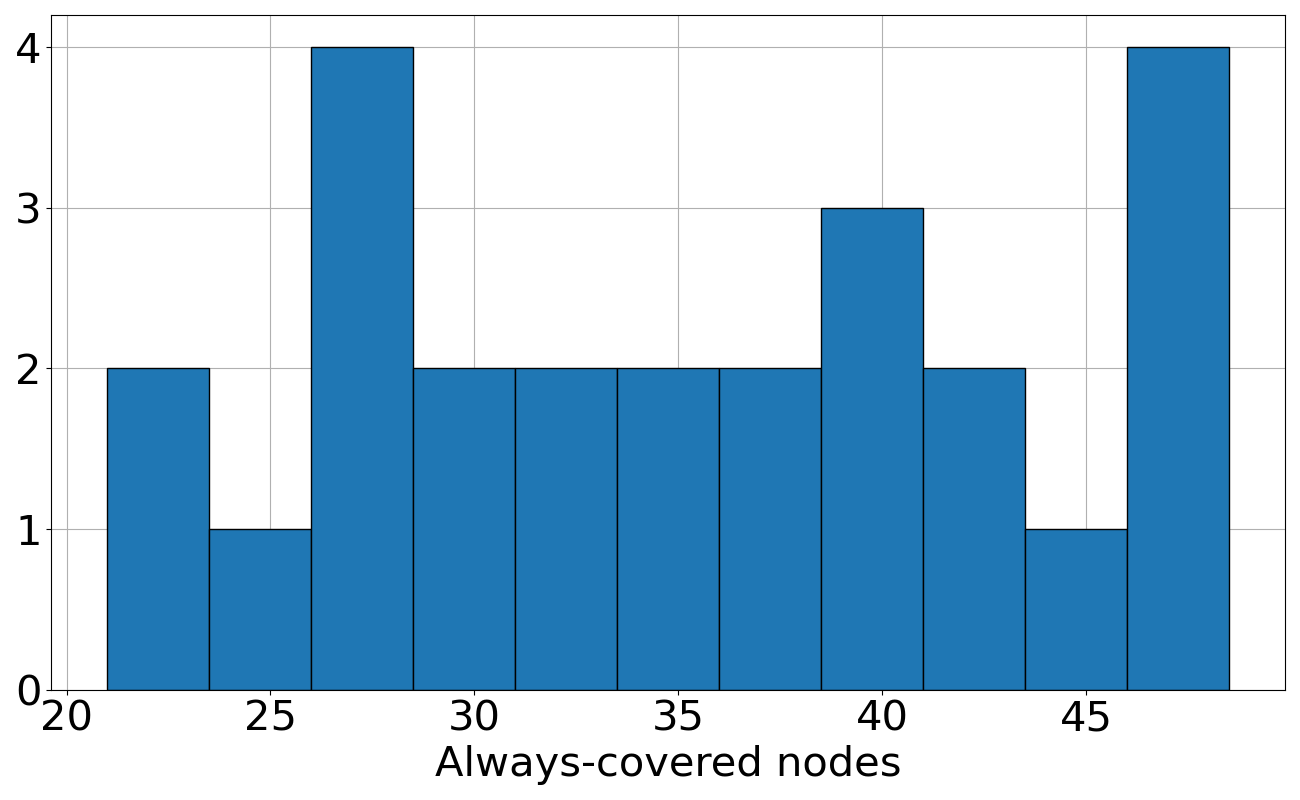}
    \end{minipage}
    \begin{minipage}{0.3\textwidth}
        \centering
        \includegraphics[width=0.9\textwidth]{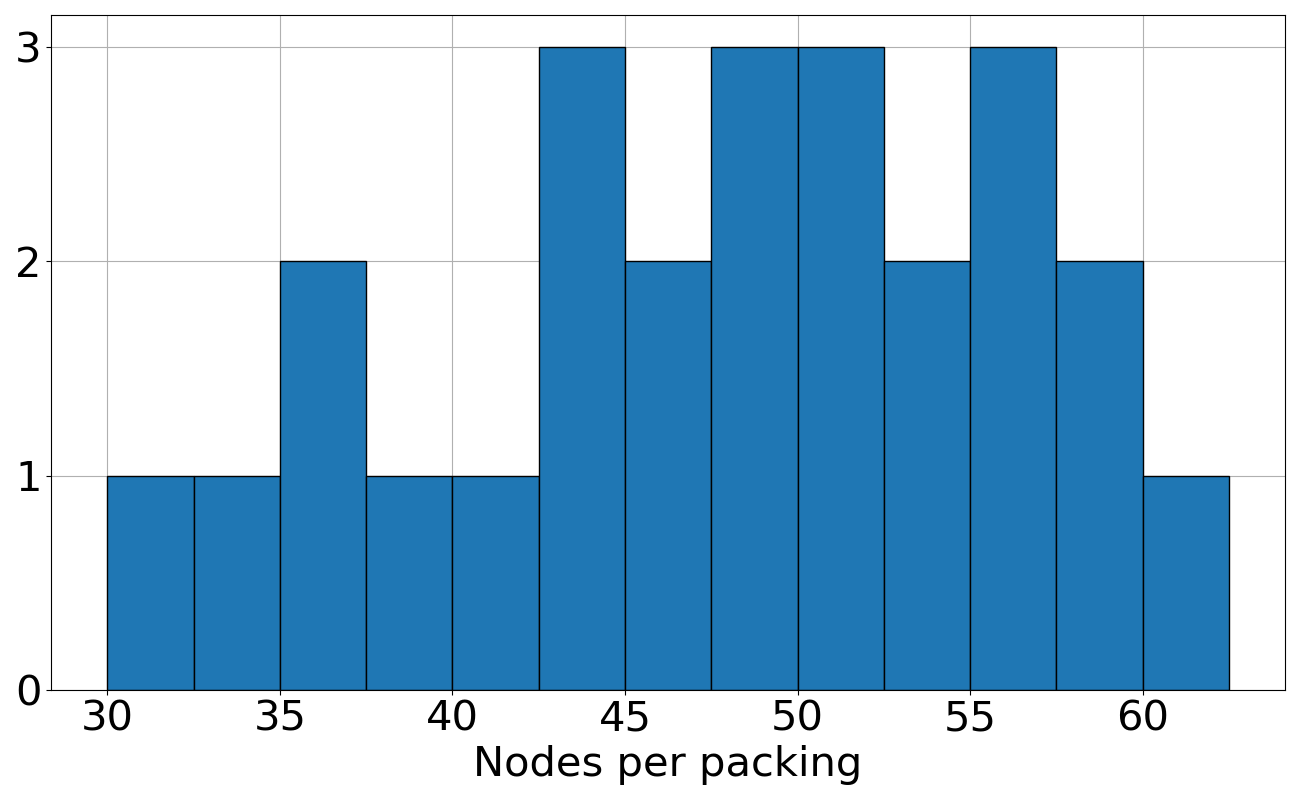}
    \end{minipage}
    \caption{Histograms of three coverage-related metrics over the KEP1 instances.}
\label{fig:two-three-hists}
\end{figure}
In addition to the aggregated results presented in KEP1 results in \autoref{fig:two-three}, it is interesting to see the extent of variation among the 25 inputs. These are given in \autoref{fig:each-instance}, where again for {\sf Implicit} and {\sf Heuristic}, we report on the confidence intervals obtained from estimating the probability of inclusion by repeated trials of that randomized algorithm (in contrast to {\sf Leximin}, where they are exact values determined by the optimal lottery computation). 
\begin{figure}[htbp!]
\centering
\includegraphics[width=6in]{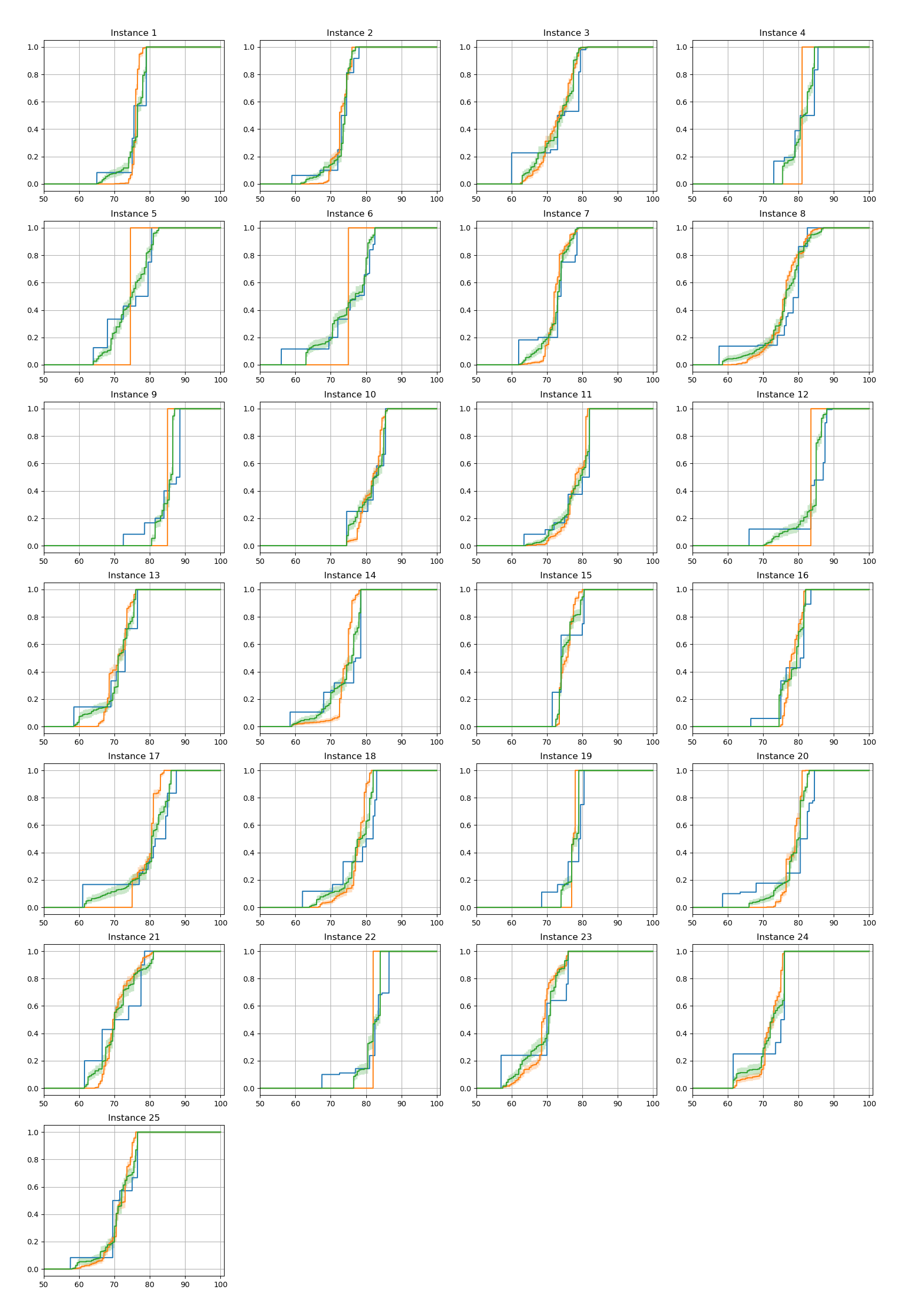}
\caption{Inclusion probabilities of {\sf Leximin}, {\sf Heuristic}, and {\sf Implicit} (using the same color scheme as before) for maximum-cardinality 2- and 3-cycle packings on all 25 KEP1 instances.}
\label{fig:each-instance}
\end{figure}

The second packing class we consider has the same allowed structures and nodes-per-packing constraint as the first, but minimizes the number of 3-cycles (motivated by the greater logistical difficulty of 3 simultaneous transplants rather than 2). Results are shown in \autoref{fig:rel-card-and-min-3-cycs}.

From \autoref{prop:coverage-loss} we have that all nodes in at least one 2-cycle or 3-cycle can be covered if we relax the nodes-per-packing constraint so that packings of size at most 3 less than optimal are allowed. However, for a particular instance, it may be possible to cover all nodes with a smaller reduction in packing size. Therefore, we define our third packing class to be 2-cycle and 3-cycle packings of cardinality at most $\delta^*$ less than optimal, where $\delta^*$ is the minimum such value required to cover all nodes in at least one 2-cycle or 3-cycle with non-zero probability. The value $\delta^*$ can be obtained by iteratively incrementing the cardinality constraint and optimizing the maximin objective until it has a nonzero value. In comparing the results in the left plots of \autoref{fig:two-three} and \autoref{fig:rel-card-and-min-3-cycs}, it is striking to observe the extent to which minimizing the number of 3-cycles in the allowed packings has only a minimal effect on {\sf Leximin}, with perhaps a somewhat greater effect on the {\sf Heuristic} algorithm. Furthermore, the center plot in \autoref{fig:rel-card-and-min-3-cycs} shows that significantly fairer lotteries result by relaxing the maximum-cardinality constraint; one key reason for this is evident from the right plot in that figure, since in essentially half of the instances, it suffices to use packings that omit at most one additional node.

\begin{figure}[htbp!]
\centering
\begin{minipage}{0.3\textwidth}
        \centering
        \includegraphics[width=0.9\textwidth]{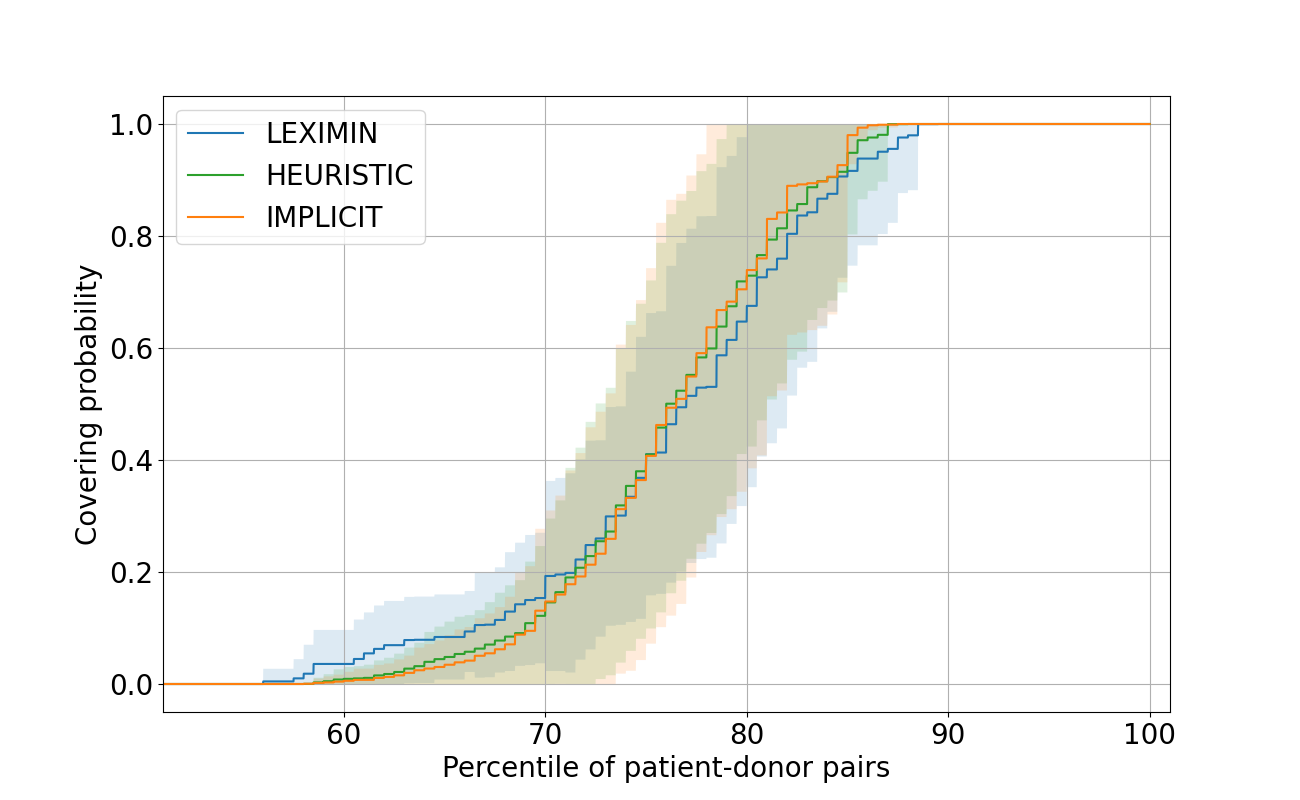}
    \end{minipage}
    \begin{minipage}{0.3\textwidth}
        \centering
        \includegraphics[width=0.9\textwidth]{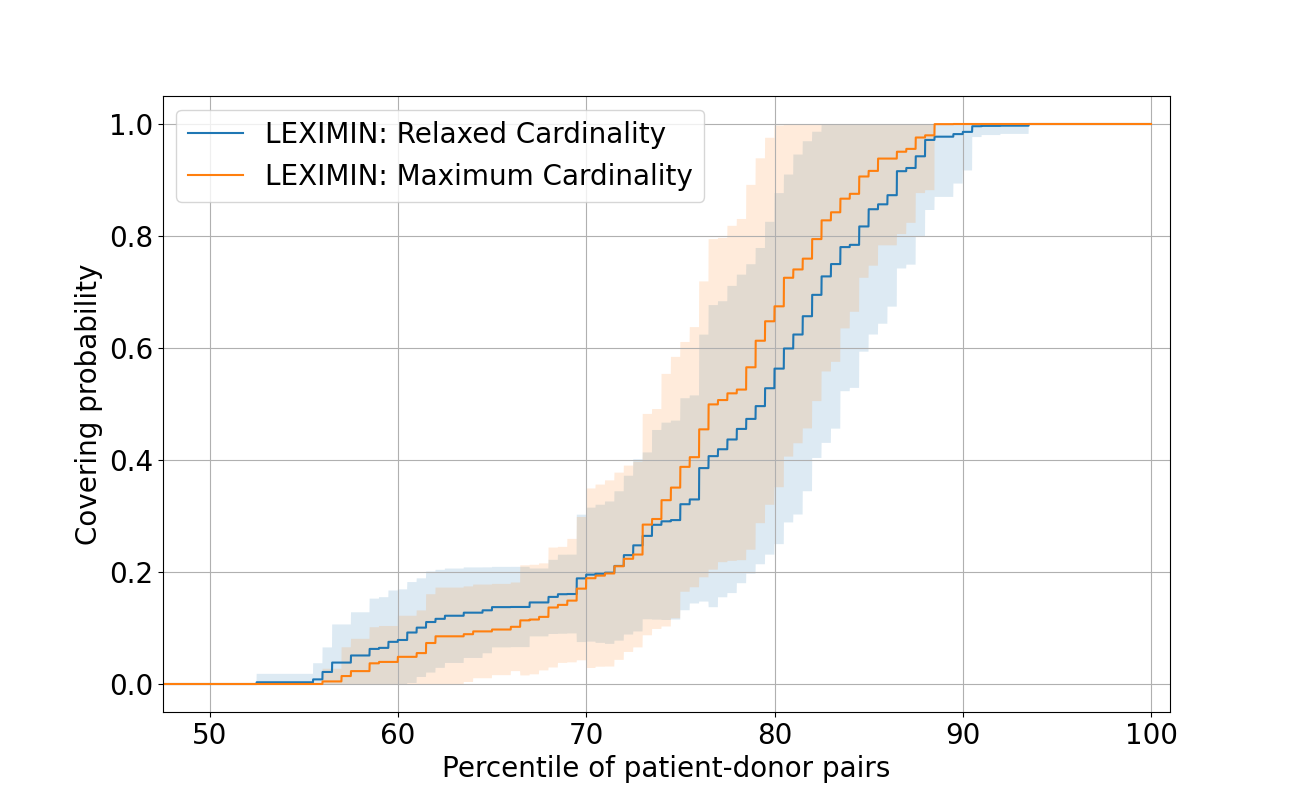}
    \end{minipage}
    \begin{minipage}{0.3\textwidth}
        \centering
        \includegraphics[width=0.8\textwidth]{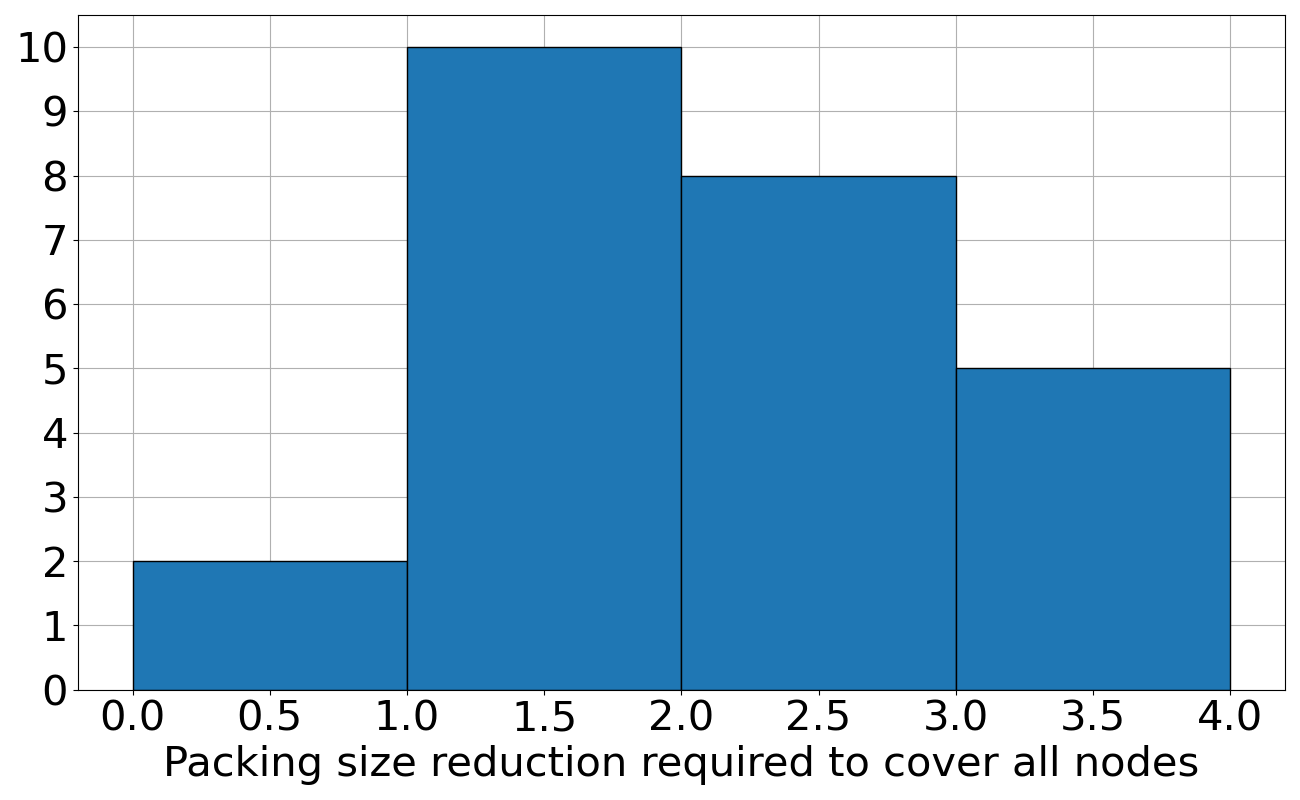}
    \end{minipage}
    \caption{(Left) Randomized algorithm results for maximum-cardinality 2-cycle and 3-cycle packings using the minimum number of 3-cycles. (Center) {\sf Leximin} inclusion probabilities for 2-cycle and 3-cycle packings of size at most $\delta^*$ less than optimal, with maximum-cardinality results included for comparison. (Right) Histogram of $\delta^*$ over all instances. All results are for KEP1.}
\label{fig:rel-card-and-min-3-cycs}
\end{figure}

Next, we turn our attention to maximum-cardinality matchings and maximum-weight matching. In this setting, the comparative performance of {\sf Leximin}, {\sf Heuristic}, and {\sf Implicit} as shown in \autoref{fig:matchings_three_ways}, is consistent with the results detailed above where 3-cycles are allowed. One point of note is that we observe little difference when comparing {\sf Leximin} on maximum-weight matchings, and on maximum-weight matchings of maximum cardinality. Although for our experiments, we have implemented these with the same ILP-based infrastructure, we could have instead used purely combinatorial algorithms. 

The comparison between the three principal packing structures is also striking: for maximum-cardinality matching, the mean number of covered nodes is 35.12, whereas for maximum-cardinality 2-cycle/3-cycle packings, the mean number of covered nodes is 47.48.

\begin{figure}[htbp!]
\centering
\begin{minipage}{0.3\textwidth}
        \centering
        \includegraphics[width=0.9\textwidth]{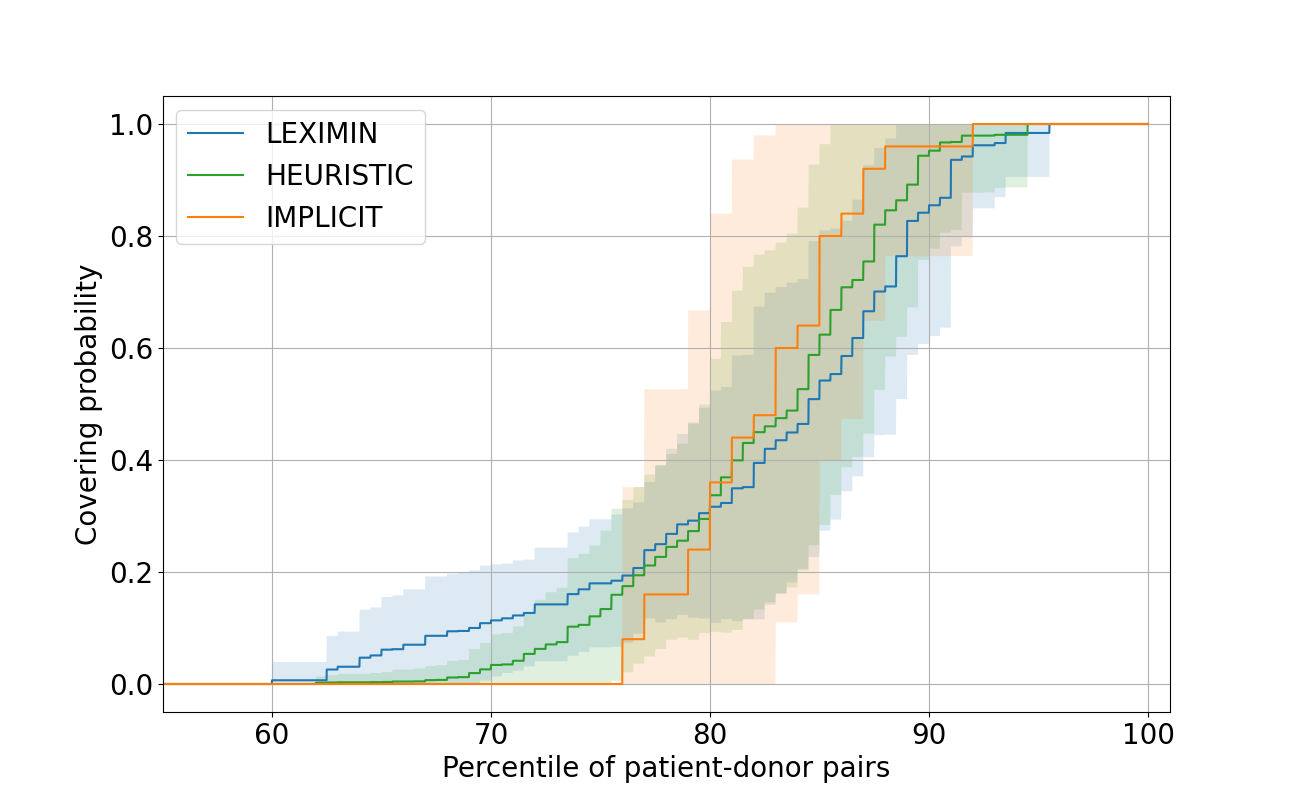}
    \end{minipage}
    \begin{minipage}{0.3\textwidth}
        \centering
        \includegraphics[width=0.9\textwidth]{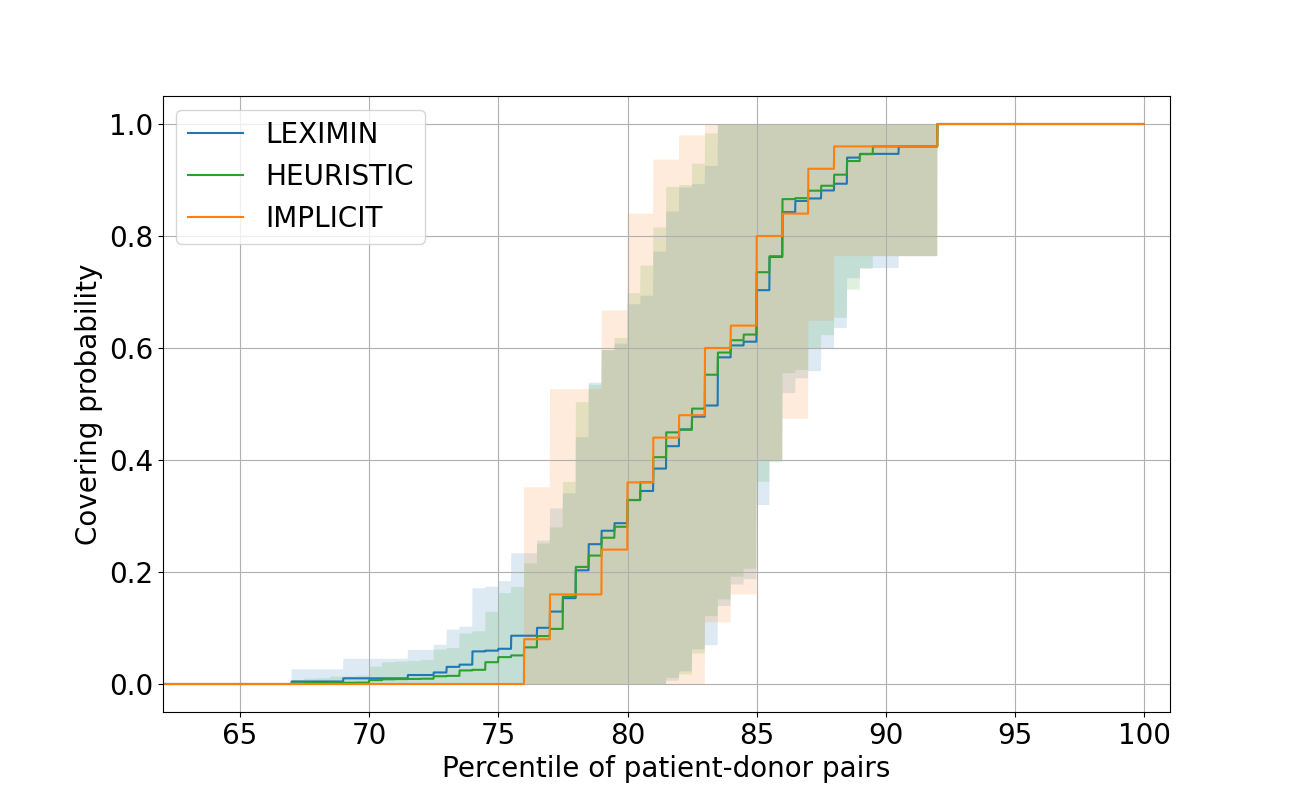}
    \end{minipage}
    \begin{minipage}{0.3\textwidth}
        \centering
        \includegraphics[width=0.9\textwidth]{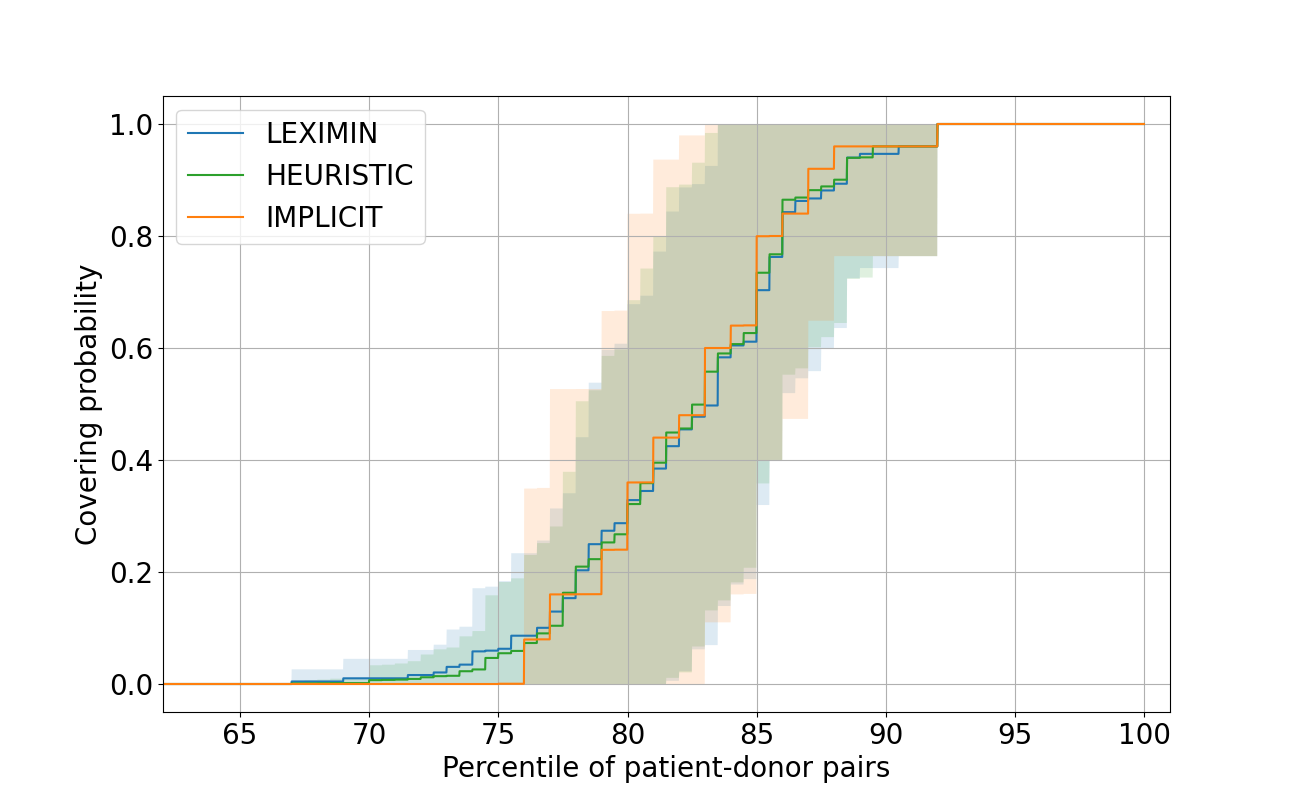}
    \end{minipage}
    \caption{Inclusion probabilities for maximum-cardinality matchings (left), maximum-weight matchings (center), and maximum-weight matchings of maximum-cardinality (right).}
\label{fig:matchings_three_ways}
\end{figure}

Kidney exchange programs may wish to bound the length of paths in their solutions because they are implemented with non-simultaneous transplants and thus provide the opportunity for patient-donor pairs to renege. Therefore, we also considered packings of 2- and 3-cycles, and paths of length at most 10 (measured by number of edges). Optimal {\sf Leximin} packing distributions in the setting with unbounded paths include packings with paths with more than 10 edges in all KEP1 instances, so this bound is nontrivial. We still find that the optimal node covering distributions with bounded chains are practically indistinguishable from those shown in \autoref{fig:unbounded-paths}, however we do find that the average path length decreases from 4.2 to 3.6 when the bound is imposed.

\subsection{Dynamic Pools}
\label{sec:dynamic-pools}
We employ a simple model of a dynamic kidney exchange pool in which pairs and NDDs arrive in batches every 30 days, and we take these batches to be the 25 KEP1 instances. Each time a batch arrives, a packing is computed and the matched pairs and NDDs leave the pool. 
30 random
orderings of batch arrivals simulate the pool evolution, and for each ordering, each randomized algorithm is run 30 times. 

Kidney exchange programs assign weights to edges in the compatibility graph based on biological properties of patients and donors as well as the amount of time the prospective recipient has already spent in the system in an effort to reduce waiting times. In our dynamic pool experiments, we adopt one such weighting function \citep{OPTN_policies} when the class of packings is matchings (i.e., 2-cycle covers) of maximum weight.

Rather than having 25 independent inputs, we view these inputs as a sequence of 25 (30-day) months of arrivals; we generate 30 random sample orderings of these 25 months. It is important to note that in each of these executions, {\sf Heuristic} in effect ``re-randomizes'' the node identities of the unmatched pairs among newly arriving pairs each month, which contributes to the performance of this approach.

In comparing the three algorithms when simulating maximum-weight matchings (with respect to the OPTN weighting function as is commonly used in practice), we obtain the  results detailed in \autoref{tab:kep1-dynamic-matching} when using the KEP1 instances.

\begin{table}[htbp!]
\centering
\begin{tabular}{|c|c|c|c|c|c|}
\hline
Algorithm & num matched &  median  &  90th percentile  &  max  & mean \\
\hline
 {\sf Implicit}         &    1243.5    &  1.0  & 11.10 & 23.53  & 3.71 \\
{\sf Heuristic}        &     1243.5    &  1.0  & 10.43 & 23.08  & 3.53 \\
{\sf Leximin}          &    1243.4    &  1.0  & 10.43 & 23.15  & 3.52 \\
\hline
\end{tabular}
\caption{Summary statistics for waiting times in KEP1 instances for maximum-weight matchings. All values are averaged over 30 orderings. The ``num matched'' column contains the total number of matched nodes, and the other columns describe the distribution of waiting times for matched nodes (in units of 30-day matching cycles)}
\label{tab:kep1-dynamic-matching}
\end{table}
 
Furthermore, in \autoref{fig:waiting-distribution-matching}, we see the full distributional variation for {\sf Leximin}.
\begin{figure}[htbp!]
\centering
\includegraphics[width=0.9\textwidth]{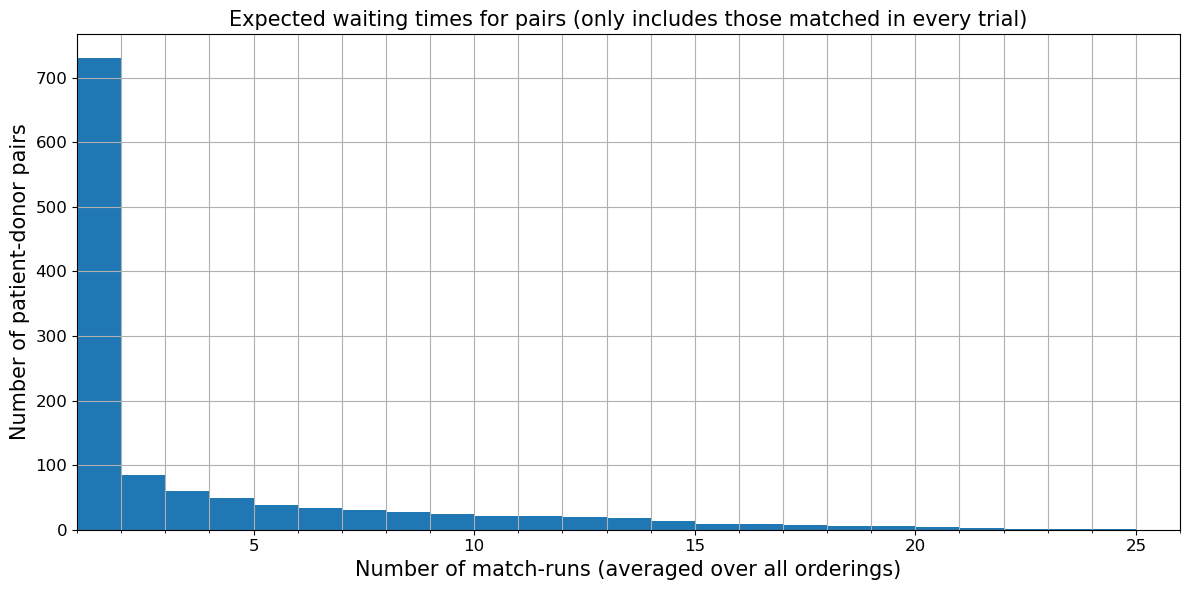}
\caption{Histogram of waiting times (in 30-day months) for weighted matchings for {\sf Leximin}. All results are for KEP1.}
\label{fig:waiting-distribution-matching}
\end{figure}

Preliminary experiments on synthetic data KEP2 (using 5 orderings of node arrivals as opposed to 30) show similar results, though due to the (unrealistic) density of those inputs, there is less significance to the differences; for example, 
waiting times increased for the average pair (92.7 days for {\sf Implicit} compared to 92.1 days for {\sf Leximin}, and for the 90th percentile wait (265.8 days for {\sf Implicit} compared to 264.3 days for {\sf Leximin}).

Finally, we also provide the analogue of these results for when the packings also allow 3-cycles in \autoref{tab:kep1-dynamic-23cycle}. Not surprisingly, this leads to substantially shorter waiting times.
\begin{table}[htbp!]
\centering
\begin{tabular}{|c|c|c|c|c|c|}
\hline
Algorithm & num matched  &  median  &  90th percentile  &  max  & mean \\
\hline
{\sf Implicit}     &    1531.3    &  1.0  & 8.03  &  23.63  & 2.94 \\ 
{\sf Heuristic} &    1531.4    & 1.0  & 7.45  &  22.99  & 2.80   \\  
{\sf Leximin}   &    1531.5    &  1.0  & 7.46  &  22.88  & 2.80  \\ 
\hline
\end{tabular}
\caption{Summary statistics for waiting times in KEP1 instances for maximum-weight 2- and 3-cycle packings.}
\label{tab:kep1-dynamic-23cycle}
\end{table}

\subsection{Implementation of {\sf Heuristic}}
\label{app:sec:heuristic-implementation}
In \autoref{sec:dynamic-pools}, we interpret {\sf {Heuristic}} as randomizing node identities as an explanation for its performance in the dynamic setting. In this section, we compare our implementation {\sf Heuristic} to one that directly enacts this interpretation by randomizing node labels before the ILP is even formulated, rather than randomizing the ILP directly. Here we provide some evidence that our results on the effectiveness of {\sf Heuristic} as a randomization mechanism still hold when implemented using this method.

\begin{figure}[htbp!]
\centering
\begin{minipage}{0.45\textwidth}
    \centering
    \includegraphics[width=0.9\textwidth]{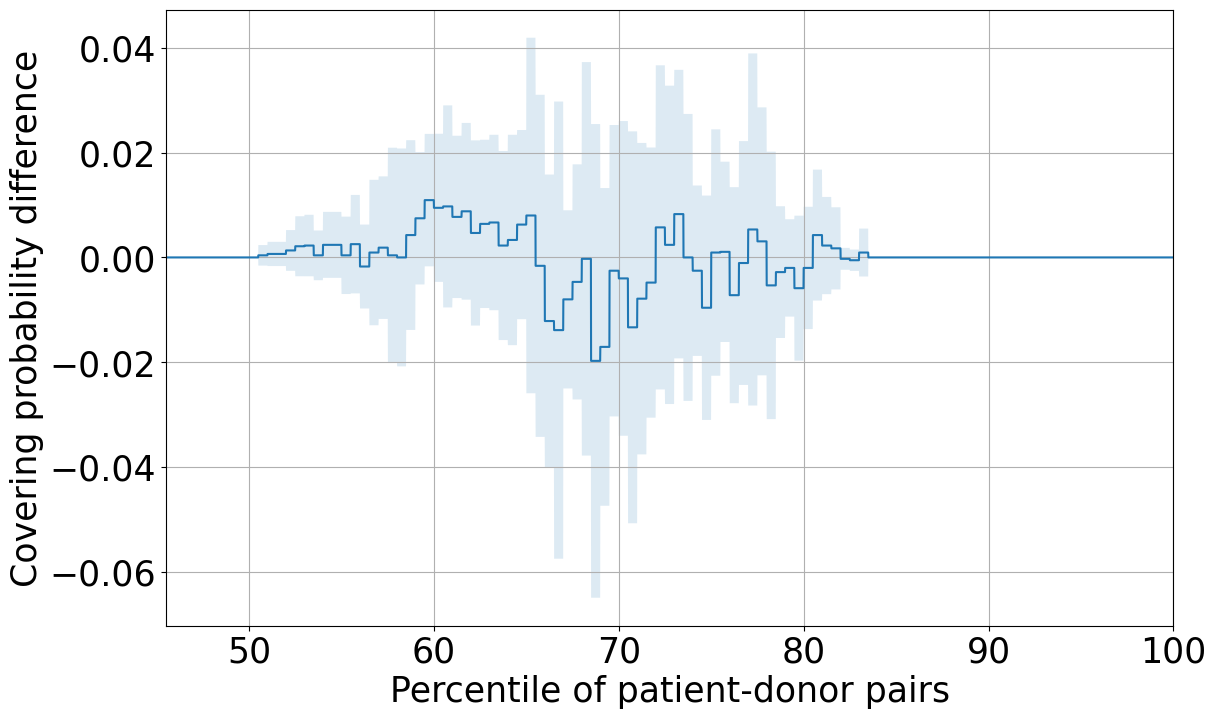}
\end{minipage}
\begin{minipage}{0.45\textwidth}
    \centering
    \includegraphics[width=0.9\textwidth]{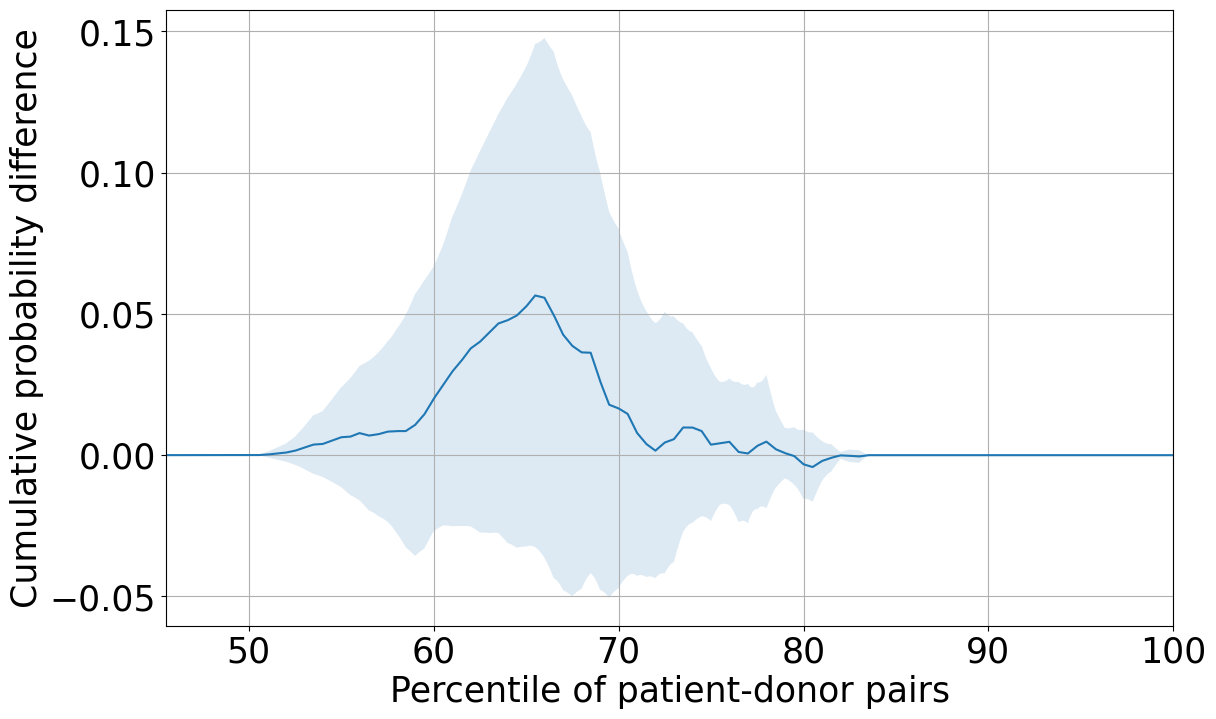}
\end{minipage}
\caption{The left plot shows average difference of node covering probability distributions between {\sf Heuristic} implemented by directly randomizing the ILP and by randomizing node identifiers. A positive number indicates a higher probability when using the former method. Results are for KEP1 using maximum-cardinality packings of 2- and 3-cycles and unbounded paths. The right plot shows the integral. Shaded region is one standard deviation.}
\label{fig:heuristics-diff}
\end{figure}

As seen in the left plot of \autoref{fig:heuristics-diff}, on average, the difference in covering probability between two nodes of the same likelihood rank between the two methods never exceeds 2\%, showing that the lotteries produced by the two implementations of  {\sf Heuristic} are very similar in terms of fairness.

Nevertheless, the right plot in \autoref{fig:heuristics-diff} shows that direct ILP randomization is, on average, the implementation which results in fairer distributions: the fact that the curve is positive for the majority of percentiles indicates that direct ILP randomization allocates more probability to low-probability nodes. This underscores the result that {\sf Leximin} significantly outperforms {\sf Heuristic} in experiments throughout this paper (Figures \ref{fig:unbounded-paths}, \ref{fig:two-three}, \ref{fig:each-instance}, \ref{fig:rel-card-and-min-3-cycs}, and \ref{fig:matchings_three_ways}), and emphasizes the benefits of explicitly optimizing for fairness.

In both the static and dynamic settings, our experiments demonstrate that without explicit randomization, the distribution of packings from which one draws solutions to a KEP can be highly unfair. By optimizing for some fairness objective, we can significantly improve the probability with which hard-to-match pairs are included in exchanges. Furthermore, even introducing randomness by shuffling variables and constraints when finding exchanges may be sufficient to achieve marked increases in waiting times in a dynamic setting. This result also demonstrates the relevance of using the fairness-optimized algorithm as a benchmark when evaluating the heuristic’s performance.

\section*{Acknowledgements}
We would like to express our sincere gratitude to Itai Ashlagi for facilitating the computational experiments done in the paper.

\bibliographystyle{plainnat}
\bibliography{matching}

\end{document}